\documentclass[envcountsame,runningheads]{llncs}
\usepackage{amsmath,amsfonts,amssymb}
\makeatletter
\newcommand{\eqnum}{\refstepcounter{equation}\textup{\tagform@{\theequation}}}
\makeatother

\usepackage{thm-restate}
\usepackage{diagrams}
\usepackage[only,llbracket,rrbracket,bigsqcap,rrparenthesis,llparenthesis,leftrightarroweq]{stmaryrd}
\usepackage{enumitem}

\usepackage{tikz}
\usetikzlibrary{automata, positioning, shapes, arrows, calc}
\SetSymbolFont{stmry}{bold}{U}{stmry}{m}{n}

\newarrow{Map}---->
\newarrow{ExistsMap}{dash}{dash}{}{dash}>
\newarrow{Embedding}C-+->
\newarrow{Mono}C--->
\newarrow{Epi}----{>>}
\newarrow{Open}o---{>>}
\newarrow{Closed}{triangle}--->
\newarrow{EMap}=====
\newcommand{\cat}[1]{\ensuremath{\mathbf{#1}}}
\newcommand{\set}{\cat{Set}}

\newcommand{\op}[1]{{#1}^{\text{op}}}
\newcommand{\ccat}{\cat{Cat}}
\newcommand{\elements}[1]{\mathbb{E}(#1)}
\newcommand{\contrapower}{\hat{\mathcal P}}
\newcommand{\power}{\mathcal P}
\newcommand{\fpower}{\mathcal {P}_\omega}
\newcommand{\forget}[1]{|#1|}
\newcommand{\klcat}[1]{\cat{Kl}(#1)}
\newcommand{\act}{\mathit{Act}}
\newcommand{\setrel}{\cat{Rel}}
\newcommand{\inverse}[1]{{#1}^{-1}}
\newcommand{\coalg}[2]{\mathbf{Coalg}_{#2}(#1)}
\newcommand{\step}[1]{\xrightarrow {#1}}
\newcommand{\eseq}{\varepsilon}
\newcommand{\id}[1]{\text{id}_{#1}}
\usepackage{bbm}
\newcommand{\1}{\mathbbm {1}}

\newcommand{\Eq}{\mathrm{Eq}}

\newcommand{\Qu}{\mathbb {Q}}
\newcommand{\sfunctor}{\mathcal S}
\newcommand{\tfunctor}{\mathcal T}
\newcommand{\semantics}[1]{\llbracket #1 \rrbracket}
\newcommand{\theory}[1]{\llparenthesis #1 \rrparenthesis}
\newcommand{\reflector}{\mathfrak {r}}
\newcommand{\parrows}[2]{\pile{\rTo^{#1}\\ \rTo_{#2}}}
\newcommand{\condset}{\mathbb K}

\newcommand{\coklcat}[1]{\cat{coKl}(#1)}
\newcommand{\bisim}{\leftrightarroweq}
\newcommand{\ba}{\cat{BA}}

\newcommand{\emcat}[1]{\cat{EM}(#1)}

\newcommand{\congset}[1]{\mathrm{Cong}(#1)}

\newcommand{\term}{\downarrow}
\newcommand{\fseq}[1]{{#1}^{\star}}
\makeatletter
\providecommand*{\twoheadrightarrowfill@}{%
  \arrowfill@\relbar\relbar\twoheadrightarrow
}
\providecommand*{\xtwoheadrightarrow}[2][]{%
  \ext@arrow 0359\twoheadrightarrowfill@{#1}{#2}%
}

\makeatother
\newcommand{\nda}[1]{\act\times{#1} + 1}
\newcommand{\mats}{\mathcal {M}_\mathbb F}
\newcommand{\trace}{\mathbf{tr}}
\newcommand{\veccat}{\mathbf{Vect}}

\usepackage{proof}
\usepackage[colorinlistoftodos,textsize=tiny,color=orange!70]{todonotes}
\usepackage{wrapfig}

\begin{document}
\title{Predicate and relation liftings for coalgebras with side effects: an application in coalgebraic modal logic}
%
\titlerunning{Predicate and relation liftings}
%
\author{Harsh Beohar \inst{1} \and
Barbara K\"{o}nig\inst{2} \and
Sebastian K\"{u}pper\inst{3} \and
Christina Mika-Michalski\inst{4}}
\authorrunning{H. Beohar et al.}
%
\institute{The University of Sheffield, U.K.\\
\email{h.beohar@sheffield.ac.uk} \and
Universit\"{a}t Duisburg-Essen, Germany \\
\email{barbara\_koenig@uni-due.de} \and
FernUniversit\"{a}t Hagen, Germany \\
\email{sebastian.kuepper@fernuni-hagen.de} \and
Hochschule Rhein-Waal, Germany \\
\email{christina.mika-michalski@hochschule-rhein-waal.de} }
\maketitle              
\begin{abstract}
  We study coalgebraic modal logic to characterise
  behavioural equivalence in the presence of side effects, i.e., when
  coalgebras live in a (co)Kleisli or an Eilenberg-Moore category. Our
  aim is to develop a general framework based on indexed
  categories/fibrations that is common to the
  aforementioned categories. In particular, we show how the
  coalgebraic notion of behavioural equivalence arises from a relation
  lifting (a special kind of indexed morphism) and we give a general
  recipe to construct such liftings in the above three cases. Lastly,
  we apply this framework to derive 
  logical
  characterisations for (weighted) language equivalence and
  conditional bisimilarity.
\end{abstract}

\section{Introduction}\label{sec:intro}


Coalgebra \cite{Rut03:universal} offers a categorical framework for
specifying and reasoning about state-based transition systems in a
generic way. In particular, new types of transition systems,
behavioural equivalences (or distances), modal logics and games can be
obtained by suitably instantiating the theory of coalgebras. While
many types of transition systems can already be studied in the
category $\set$, systems with side effects -- leading to a notion of
trace equivalence or conditional bisimilarity -- usually require to
move to a setting beyond $\set$, using Kleisli, coKleisli or
Eilenberg-Moore categories, where the (co)monad specifies the
side-effects.

Behavioural equivalences for such scenarios have already been studied
extensively (see e.g. \cite{ABHKMS12,Hasuo2007_GenericTraceSemantics,jacobs:trace_semantics,Jacobs_etal:trace_EM_KL}). Modal logics, on the other hand, have been considered to a lesser extent with side effects \cite{kk:generic-trace-logics,forgetfullogics,dorsch_et_al:LIPIcs:2019:10938}; the emphasis (see the survey \cite{CoalgebraicLogicSurvey_KupkePattinson2011} and the recent articles \cite{GorinSchroeder2013:Lambdabisim,kk:generic-trace-logics,Katsumata2021_QuantitativeModLogic,KupkeRot_ExprLog_2020}) has been on logical characterisation of various notions of bisimulation relations and metrics.
%
The aim of the present paper is to close this gap 
by applying the dual adjunction setup for fibrations developed by Kupke and Rot
\cite{KupkeRot_ExprLog_2020} to derive logical characterisations for 
coalgebras with side effects. Recently, another powerful approach
\cite{dorsch_et_al:LIPIcs:2019:10938} based on graded monads has been
developed to handle equivalences from the van Glabbeek spectrum and
beyond. There the characterising notion is not coalgebraic behavioural
equivalence, but a refinement of it, called finite depth behavioural
equivalence. In addition, Kupke and Rot \cite{KupkeRot_ExprLog_2020}
when comparing their work with \cite{dorsch_et_al:LIPIcs:2019:10938}
noted that ``trace equivalences of various kinds covered in
\cite{dorsch_et_al:LIPIcs:2019:10938} cannot be captured directly in
their setup''. Hence we will show how to capture linear notions such as trace, language, and failure equivalences in their setup.

\begin{wrapfigure}[5]{r}{0.45\textwidth}
\vspace{-0.9cm}
\begin{tikzpicture}
\node (e) at (-1,2.5) {$\mathbf E$};
\node (c) at (-1,1) {$\mathbf {Set}$};
\node (l) at (-2.25,1.75) {\eqnum\label{eq:idea}};
\node (ee) at (-4,2.5) {$\mathbf{Coalg}_{\mathbf E} (F_\lambda)$};
\node (cc) at (-4,1) {$\mathbf{Coalg}_{\mathbf {Set}}(F)$};
\path[->,thick]
(ee) edge (e)
(ee) edge (cc)
(e) edge (c)
(cc) edge (c)
(c) edge[loop right] node {$F$} (c)
(e) edge[loop right] node {$F_\lambda$} (e);
\end{tikzpicture}
\end{wrapfigure}
Here, in order to treat these linear notions of equivalence uniformly, we follow the approach of Hermida and Jacobs \cite{Hermida_Jacobs:str_coind_fibration}
to capture bisimulation relations using the language of fibrations \cite{jacobs-fibrations};
they have increasingly appeared in the coalgebraic literature
\cite{Bonchi_et_al:UpToFibrations,BonchiEtal_CoinductionFibration,hasuo_kataoka_cho_2018,Jacobs_indexedcat,Klin2005:LeastFibredLifting,CodensityGames,KupkeRot_ExprLog_2020,SprungerEtal_FibBisimQuantitative}.
%
%
The general idea is as follows and is illustrated above in
\eqref{eq:idea}. First, a system is modelled as a coalgebra $X
\rTo^\alpha FX$ for some endofunctor $\set \rTo^F \set$. Second a
fibration $\cat E$ of binary relations on the working category of sets
is realised, whose fibres are all the relations on the underlying
state space. 
Third,
a mechanism $\power(X\times X) \rTo^\lambda \power(FX \times FX)$ (aka
relation lifting) 
is
defined, which amounts to the lifting of $F$ to an endofunctor
$F_\lambda$ on $\cat E$. Now one can study the coalgebras induced by
$F_\lambda$ and, more importantly, this category
$\coalg{F_\lambda}{\cat E}$ can be again arranged (see
\eqref{eq:idea}) as a fibration on $\coalg{F}{\set}$. Lastly, the
applicability is shown by characterising bisimulation relations
on $X$ as coalgebras of a certain endofunctor living in the fibre above $(X,\alpha)$.

One of the objectives of this paper is to extend this `categorical' picture \eqref{eq:idea} w.r.t.\thinspace coalgebraic notion of behavioural equivalence for dynamical systems having side effects, i.e., those systems that can be modelled as coalgebras living in a (co)Kleisli or Eilenberg-Moore category for some (co)monad on $\set$. Typical examples, in the context of this paper, are the following: nondeterministic (linear weighted) automata modelled as coalgebras in the Kleisli category for the powerset (multiset) monad (see Sections~\ref{sec:KleisliEx:NDA} and \ref{sec:KleisliEx:LWA}); conditional transition systems (which facilitate formal modelling of software product lines) modelled as coalgebas in the coKleisli category for the writer comonad $\condset\times\_$ (see Section~\ref{sec:CTS}). We left out case studies in Eilenberg-Moore categories in this paper; however, they are shown to satisfy the assumptions of this paper (Section~\ref{sec:EMcats}).

We explore the general conditions on relation liftings (technically
they are called indexed morphisms in the paper like the ones indicated
by $\lambda$ in \eqref{eq:idea}) to ensure that behavioural
equivalences
can be viewed as coalgebras living in the fibre above a given
coalgebra $(X,\alpha)$ with side effects. Another contribution is a
recipe for obtaining 
relation 
and predicate liftings (a special type of indexed morphisms) whose
definition and correctness proof are otherwise (at least in the
Kleisli case) quite cumbersome to establish. 
Predicate liftings are instrumental in providing interpretation to
various modalities for coalgebras in (co)Kleisli or Eilenberg-Moore
categories, just like in the case of $\set$ (cf.\thinspace
\cite{Jacobs_indexedcat,PATTINSON2003177,SCHRODER2008230}).
Technically, our study focuses on lifting an indexed morphism for a given endofunctor $F$ on $\set$ to an indexed morphism for a (co)Kleisli extension/Eilenberg-Moore lifting $\bar F$ of
$F$. 
And to the best of our knowledge, this question is open at least for
coalgebras with side effects.  

Once we have captured behavioural equivalence in a fibration, we can
then apply the Kupke-Rot setup \cite{KupkeRot_ExprLog_2020} based on
dual adjunctions (see the survey
\cite{CoalgebraicLogicSurvey_KupkePattinson2011} on coalgebraic modal
logic) to establish the logical characterisation of behavioural
equivalence. In particular, we first construct the Kupke-Rot setup for
behavioural equivalences 
and show that the sufficient conditions for adequacy (i.e.,
behavioural equivalence is contained in logical equivalence) and
expressivity (i.e., converse of adequacy) given in
\cite{KupkeRot_ExprLog_2020} are satisfied. This setup is later used
to derive the logical characterisation for (weighted) language
equivalence and conditional bisimilarity; note that these notions were
not studied in \cite{KupkeRot_ExprLog_2020}.

While several ingredients (especially encompassing fibrations) used in this paper are already known, our
paper contains the following original contributions:
\begin{itemize}
  \item We capture behavioural equivalences on coalgebras beyond $\set$ as a fibred notion by characterising them 
      as special types of coalgebras.
  \item We give concrete recipes for defining predicate
     and relation liftings (which is both tedious and error-prone) in (co)Kleisli and Eilenberg-Moore categories.
  \item We extend the dual adjunction framework for fibrations by Kupke and Rot to side effects, in particular to Kleisli categories. Here we need a mechanism to factor 
      the state space of a coalgebra by behavioural equivalence, which is difficult if the category has no coequalisers. We provide a technique based on reflective subcategories to circumvent this issue.
\end{itemize}
This paper is organised as follows. Section~\ref{sec:prelim} sets the
relevant categorical preliminaries required for this paper. It is
assumed that the reader is already familiar with basic category
theory, particularly, how a Kleisli or an Eilenberg-Moore category is
induced by a monad. Section~\ref{sec:rel-lift} introduces the
assumptions that ensure behavioural equivalence is a fibred notion (in
the sense of \eqref{eq:idea}). Section~\ref{sec:CoalgModalLog} is
devoted to coalgebraic modal logic where general adequacy and
expressivity results for behavioural equivalence are derived from
\cite{KupkeRot_ExprLog_2020}. Section~\ref{sec:indMorKleisli} gives
the recipe to construct relation/predicate liftings for coalgebras with side effects. 
In the next sections, the results of this paper are applied in the
context of nondeterministic automata 
and conditional transition systems. Section~\ref{sec:conclusion}
concludes this paper with some discussions on future research. Note
that proofs as well as additional material
on linear weighted and generalised Moore automata in Kleisli and Eilenberg-Moore categories can be found in the appendix. 

\section{Preliminaries}\label{sec:prelim}
\paragraph*{Coalgebraic preliminaries \cite{jacobs2016:coalg:book,Rut03:universal}}
Let $\cat C$ be a category and let $\cat C \rTo^F \cat C$ be an endofunctor modelling the branching type of the system of interest. Then the behaviour of a state-based system will be modelled as an $F$-coalgebra (or, simply coalgebra), i.e., as a morphism $X \rTo^\alpha FX$ in the category $\cat C$.

\begin{definition}
  A \emph{coalgebra homomorphism} $f$ between $(X,\alpha)$ and $(Y ,\beta)$ is a morphism $X \rTo^f Y\in\cat C$ satisfying $Ff \circ \alpha = \beta \circ f$. The collection of coalgebras and their homomorphisms forms a category denoted $\coalg F {\cat C}$.
\end{definition}
Moreover, one can define behavioural equivalence on the (concrete) states of a coalgebra under the assumption that there is a functor $\cat C \rTo^{\forget{\_}} \set$. By the concrete state-space of a coalgebra $(X,\alpha)$, we mean the set $\forget X$. Typically, in our case studies, the functor $\forget{\_}$ will be a forgetful functor and will have a left/right adjoint $\iota$. For instance, $\iota$ is left adjoint to $\forget{\_}$ when $\cat C=\klcat{T}$ or $\cat C=\emcat T$ (for some monad $\set \rTo^T \set$), while it is right adjoint when $\cat C=\coklcat G$ for some comonad $\set \rTo^G \set$. 
\begin{definition}\label{def:behequiv}
  Two states $x,x'\in\forget{X}$ of a coalgebra $(X,\alpha)$ are \emph{behaviourally equivalent} iff there is a coalgebra homomorphism $f$ such that $\forget{f} x = \forget{f} x'$.
\end{definition}
\begin{example}\label{ex:nda}
  An interesting example of coalgebras living in Kleisli categories is nondeterministic automata (NDA).
  Following \cite{Hasuo2007_GenericTraceSemantics} an NDA is a coalgebra living in
  $\cat C = \klcat{\power}$, which is isomorphic to the category $\setrel$ of sets as objects
  and relations as maps. Recall that a Kleisli extension $\setrel \rTo^{\bar F} \setrel$ of $\set \rTo^F \set$ (i.e. $\bar F\circ \iota = \iota \circ F$) is in correspondence \cite[Theorem~2.2]{m:lifting-kleisli}  with a distributive law $FT \rTo^\vartheta TF$ such that the following diagrams commute in $\set$.

  \vspace{-0.6cm}
  \begin{diagram}[width=2em,height=2.5em]
    FX & &\rEMap && FX &&& FTTX & \rTo^{\vartheta_{TX}} & TFTX & \rTo^{T\vartheta_X} & TTFX\\
    \dTo^{F\eta_{X}} &&&& \dTo_{\eta_{FX}} &&& \dTo^{F\mu_X} & & & & \dTo_{\mu_{FX}} && \eqnum\label{eq:KLlaw}\\
    FTX & & \rTo^{\vartheta_X}& & TFX &&& FTX & & \rTo^{\vartheta_X} && TFX
  \end{diagram}
  Consider $F\_ = A\times\_ +1$ (where $1 = \{\bullet\}$) with the following distributive law \cite{jacobs:trace_semantics}:

  \vspace{-0.4cm}
  \begin{equation}\label{eq:thetaNDA}
    \act\times \power X + 1 \rTo^{\vartheta_X} \power (\act \times X +
    1) \qquad\quad (a,U) \mapsto \{a\}\times U, \ \bullet \mapsto
    \{\bullet\}.
  \end{equation}
  This induces a functor $\setrel\rTo^{\bar{F}}\setrel$
  which acts on a relation $X\rTo^f Y$, seen as a Kleisli arrow
  $X\rTo^{f'}\power{Y}$, as follows:
  $\bar{F}f = \vartheta_Y\circ Ff'$. Notice that
  $\bar{F}$-coalgebras model implicit nondeterminism (i.e. this side-effect is hidden to an outside observer) \cite{Hasuo2007_GenericTraceSemantics}, thus behavioural equivalence typically coincides with \emph{language equivalence} (instead of bisimilarity) in this case.
\end{example}

\vspace{-0.5cm}
\paragraph*{Predicate liftings as indexed morphisms}
Predicates and their liftings are quite common within the literature on (coalgebraic) modal logic. In particular, a predicate is used as the semantics of a logical formula \cite{PATTINSON2003177}, or as a relation on the state space of a coalgebra \cite{Hermida_Jacobs:str_coind_fibration}.
In the basic setting, when $\cat C=\set$, the predicates on a set $X$ are given by the subsets of $X$. Now, given a function $X \rTo^f Y$, a predicate $V$ on $Y$ (i.e.\thinspace $V \subseteq Y$) can be transformed into a predicate on $X$ by the pullback operation $\inverse{f} V\subseteq X$ in $\set$. Note that this operation is functorial in nature; thus this `logical' structure can be organised as a functor $\op{\set} \rTo^\contrapower \ccat$ \cite{Jacobs_indexedcat}, where $\contrapower X$ is the poset $(\power X,\subseteq)$ viewed as a category. As noted by Jacobs in \cite{Jacobs_indexedcat},
\emph{predicate logic on a category is given by an indexed category}
and \emph{predicate liftings are (endo)morphisms of indexed categories}.
\begin{definition}
  \label{def:indexed_cat}
  An indexed category is a $\ccat$-valued presheaf, i.e., a
  contravariant functor $\Phi$ from $\cat C$ to $\ccat$. In addition,
  a morphism between two indexed categories
  $\op{\cat C} \rTo^\Phi \ccat$ and $\op{\cat D} \rTo^\Psi \ccat$ is a  pair of a functor $\cat C \rTo^G \cat D$ and a natural
  transformation $\Phi \rTo^\lambda \Psi \circ \op G$.
\end{definition}
\begin{note}
  Often 
  the application of $\Phi$ on $f\in\cat C$ is denoted as $f^*$. We also omit the use of superscript `$\text{op}$'
  on functors 
  and
  use the phrases `indexed morphism' and `predicate lifting' interchangeably. 
\end{note}
\begin{remark}
  Another, equivalent, way to organise logic is by specifying the fibration of predicates over a category \cite{jacobs-fibrations}. The transformation of a fibration over $\cat C$ into a contravariant pseudofunctor $\cat C \rTo\ccat$ is given by taking the fibres at each object in $\cat C$. Conversely one has to invoke the so-called \emph{Grothendieck construction} to get a fibration, which glues all the fibres $(\Phi X)_{X\in\cat C}$ to form a total category of predicates $\elements \Phi$ defined as follows.

  $
  \infer={(X,U) \in \elements\Phi}{X\in \cat C \land U \in \Phi X}
  \qquad
  \infer={(X,U) \rTo^{f,\bar f} (Y,V) \in \elements{\Phi}}{X \rTo^f Y \in\cat C \land U \rTo^{\bar f} f^*V \in \Phi X}
  $

  \noindent
  Moreover, there is an obvious `forgetful' functor $\elements{\Phi} \rTo^p \cat C$ given by $(X,U) \mapsto X$ that induces a (split)   fibration on $\cat C$ \cite{jacobs-fibrations,Jacobs_indexedcat}. In the parlance of concrete categories, the functor $p$ is topological \cite[Definition~21.1]{JoyOfCats} when $\Phi$ has fibred limits. Often, in applications, the fibres $(\Phi X,\preceq)$ (at each $X\in\cat C$) form a poset (rather than a full-fledged category); we label such an indexed category/fibration as \emph{thin}. We restrict ourselves to thin fibrations in this paper. 
\end{remark}

\begin{example}\label{ex:standard_indcat}
  The contravariant powerset functor $\set \rTo^{\contrapower}
  \ccat$ is an example of an indexed category such that $\elements{\contrapower} \rTo \set$ is a bifibration \cite{jacobs-fibrations}.
  This is because the reindexing functor $\inverse{f}$ (for any function $f$) has a left adjoint given by the direct image functor $f_!$. 
  Moreover, as an example of a predicate lifting, consider
  $F=\power$ over $\cat C =\set$ (which describes the branching type
  of unlabelled transition systems) with
  $\contrapower X \rTo^{\lambda_X} \contrapower \power X$ given by $U \mapsto \power U$.
  It is well known that the above predicate lifting encodes the box
  modality $\Box$ from logic \cite{Jacobs_indexedcat}.
\end{example}

\section{Behavioural equivalence through indexed morphisms}\label{sec:rel-lift}
Indexed morphisms not only induce modalities of interest in Computer Science; but they can also be used to characterise behavioural equivalence. The original idea \cite{Hermida_Jacobs:str_coind_fibration} is to work with an indexed category $\op\set\rTo^\Psi\ccat$ of binary relations, i.e., $\Psi$ is the composition 
$\op{\set} \rTo^{\_ \times \_} \op{\set} \rTo^{\contrapower} \ccat$.
In particular, $\Psi X$ is the set of all relations on $X$. Then, for a relation lifting $\Psi X \rTo^{\lambda_X} \Psi FX$ and a coalgebra $X \rTo^\alpha FX\in\set$, bisimilarity is the largest fixpoint of the functional:
\begin{equation}\label{eq:idea:functional}
  \Psi X \rTo^{\lambda_X} \Psi FX \rTo^{\alpha^*} \Psi X.
\end{equation}
Unfortunately, this idea of working with relations on the concrete
state space immediately does not generalise to coalgebras with side
effects; e.g., in the case of conditional transition systems (CTSs)
viewed as coalgebras living in the coKleisli category of the writer
comonad $\condset \times \_$ (see
Section~\ref{sec:CTS}). The problem essentially lies in 
associating a fibre to be the set of all binary relations on the state space. There are situations (as in CTSs) where the fibres will only be some subset of all the relations on the state space.
As a result, we impose the following restriction:
\begin{enumerate}[label=\textbf{A\arabic*}]
  \item \label{assum:product} our working category $\cat C$ has binary products $\otimes$.
\end{enumerate}
Thus we can define an indexed category $\Psi$ of relations as the composition:
\begin{equation}\label{eq:indexedcatRelations}
  \op{\cat C} \rTo^{\_ \otimes \_} \op{\cat C} \rTo^{\forget{\_}} \op{\set} \rTo^{\contrapower} \ccat.
\end{equation}
We view the elements of $\Psi X$ (for some object $X\in C$) as `abstract' relations on $X$. Furthermore, \ref{assum:product} also ensures that for any object $X\in\cat C$ there is a function
\[
\forget{X \otimes X}  \rTo^{\langle\forget {\pi_1^X},\forget{\pi_2^X}\rangle} \forget{X} \times  \forget{X},
\]
where $X\otimes X \rTo^{\pi_1^X,\pi_2^X} X$ are the two projection arrows in $\cat C$. And thanks to these functions, we can define abstract equality in the fibre $\Psi X$. In particular,
\[
\equiv_X \ =\  \inverse{\langle\forget {\pi_1^X},\forget{\pi_2^X}\rangle} =_{\forget{X}}\enspace.
\]
Notice that in some cases (like when $\cat C$ is a Kleisli or Eilenberg-Moore category) the abstract equality $\equiv_X$ coincides with the equality on the concrete state space $=_{\forget{X}}$ because the forgetful functor is product preserving. However, in the context of CTSs, we will see that the two notions of equality differ.
\begin{restatable}{proposition}{EqFuncExistence}\label{prop:eqfunexists}
  Under Assumption~\ref{assum:product} the square drawn in \eqref{eq:squarepbk} commutes for any arrow $X \rTo^f Y$ in $\cat C$. As a result, there is a functor $\cat C \rTo^{\Eq} \elements{\Psi}$ (henceforth called \emph{equality} functor) that maps an object $X$ to the abstract equality $\equiv_X$.
\end{restatable}
\vspace{-0.7cm}
\begin{minipage}[t]{0.45\textwidth}
\begin{diagram}[width=3em,height=2.5em]
  \forget {X \otimes X} & \rTo{\langle\forget{\pi_1^X},\forget{\pi_2^X}\rangle} & \forget X \times \forget X\\
  \dTo^{\forget{f\otimes f}} & \eqnum\label{eq:squarepbk} & \dTo_{(\forget f \times \forget f)}\\
  \forget {Y \otimes Y} & \rTo{\langle\forget{\pi_1^Y},\forget{\pi_2^Y}\rangle} & \forget Y \times \forget Y
\end{diagram}
\end{minipage}
\begin{minipage}[t]{0.5\textwidth}
\vspace{0.2cm}
    \begin{diagram}[width=3em,height=2.5em]
    \coalg{F_\lambda}{\elements{\Phi}} = \elements{\Phi_\lambda^F} & \rTo & \elements{\Phi}\\
    \dTo^{p^F_\lambda} &\eqnum \label{eq:genresult} & \dTo_p\\
    \coalg{F}{\cat C} &\rTo & \cat C
  \end{diagram}
\end{minipage}

\vspace{0.3cm}
The next proposition (originally from \cite{Jacobs_indexedcat}) is a general result on indexed categories useful in lifting an endofunctor on $\cat C$ to an endofunctor on the given fibration $\elements{\Phi}$. Moreover the category of coalgebras of the lifted endofunctor can be structured again as a fibration on the given category of coalgebras in which our original system of interest is modelled.
\begin{proposition}\label{prop:technicalresult}
  Consider the diagram in \eqref{eq:genresult}, then the following statements hold for a given functor $\cat C \rTo^F \cat C$ and an indexed morphism $\Phi \rTo^\lambda \Phi F$.
  \begin{itemize}[noitemsep]
    \item The map $\lambda$ induces a map  $\elements{\Phi} \rTo^{F_\lambda} \elements{\Phi}$ of fibrations given by $(X,U) \mapsto (FX,\lambda_X U)$.
    \item The category of coalgebras induced by $F_\lambda$ forms a fibration on $\coalg{F}{\cat C}$, where $\op{\coalg{F}{\cat C}} \rTo^{\Phi_\lambda^F} \ccat$ is the mapping: $(X,\alpha) \mapsto \coalg{\alpha^*\circ \lambda_X}{\Phi X}$.
  \end{itemize}
\end{proposition}
Now recall \eqref{eq:idea:functional} and $\Psi$ as indexed category of relations on $\set$ (i.e. substitute $\otimes$ by $\times$ and $\forget{\_}$ by the identity functor in \eqref{eq:indexedcatRelations}), an arbitrary bisimulation relation $R$ on a coalgebra $X \rTo^\alpha FX\in\set$ is the relation $R\in\Psi X$ satisfying $R \subseteq \alpha^* \lambda_X R$. In other words, bisimulation relations on the state space $X$ are again coalgebras of the functor $\alpha^*\circ \lambda_X$ living in the fibre $\Psi X$. 
Next we show that the same holds for behavioural equivalence in general, however, 
under the following assumptions:
\begin{enumerate}[label=\textbf{A\arabic*},start=2]
  \item\label{assum:presEq} the given morphism $\Psi \rTo^\lambda \Psi F$ preserves $\Eq$, i.e., $F_\lambda \circ \Eq = \Eq \circ F_\lambda$. Equivalently, this means that $\lambda_X (\equiv_X) =\ \equiv_{FX}$ for every $X\in\cat C$. 
  \item\label{assum:quotient} the functor $\Eq$ has a left adjoint $\Qu$.
\end{enumerate}
\begin{remark}\label{rem:quotientSets}
  Assumption~\ref{assum:quotient} already 
  appeared in \cite{Hermida_Jacobs:str_coind_fibration}
  to model quotient types in the context of type theory. However, our usage is in the unit $\kappa$ of $\Qu\dashv\Eq$ to construct a witnessing coalgebra homomorphism in Theorem~\ref{thm:fibredcoalg->behequiv}. This idea is already known in type theory; for instance, see \cite[Theorem~3.7]{lmcs:738} where a similar result was proven albeit under the stronger assumption that the final coalgebra for $F$ exists. So when $\cat C=\set$, $\Qu$ maps an relation $R$ on $X$ to the quotient generated by the smallest equivalence containing $R$; the unit $\kappa_X$ (for any set $X$) is the usual quotient function mapping an element to its equivalence class.
\end{remark}
\begin{restatable}{theorem}{BehEquivIsCoalg}\label{thm:behequiv->fibredcoalg}
  Given an indexed morphism $\Psi \rTo^\lambda \Psi \circ F$,
  then under Assumptions 
  \ref{assum:product} and \ref{assum:presEq}, the behavioural equivalence induced by a coalgebra homomorphism $f\in\coalg{F}{\cat C}$ on a coalgebra $(X,\alpha)\in\coalg{F}{\cat C}$ is a $\alpha^*\circ\lambda$-coalgebra living in the fibre $\Psi X$, i.e., $ f^*(\equiv_Y) \subseteq \alpha^* \lambda_X (f^*\equiv_Y). $
\end{restatable}
\begin{restatable}{theorem}{FibCoalgIndBehEquiv}\label{thm:fibredcoalg->behequiv}
Under Assumptions~
\ref{assum:product}, \ref{assum:presEq}, and \ref{assum:quotient} for
    every $\alpha^*\circ\lambda$-coalgebra $R$ there is a coalgebra homomorphism $f\in\coalg{F}{\cat C}$ such that $R\subseteq f^*(\equiv_{\text{cod}(f)})$, where $\text{cod}(f)$ denotes the codomain of $f$. Moreover, $R=f^*(\equiv_{\text{cod}(f)})$ when the unit of $\Qu\dashv \Eq$ is Cartesian.
\end{restatable}
\begin{remark}
  An application of Theorem~\ref{thm:fibredcoalg->behequiv}
  could be in establishing the completeness of coalgebraic games (as
  in the spirit of \cite{km:bisim-games-logics-metric}). 
  For instance, if the winning positions of
  Duplicator viewed as a relation $R$ is a coalgebra in the fibre of
  $\Psi$, then Theorem~\ref{thm:fibredcoalg->behequiv} can be used to
  show that winning positions of Duplicator are behaviourally
  equivalent. In the future, we would like to test this application by
  working out a notion of 2-player games for coalgebra with side effects.
\end{remark}
\section{Coalgebraic modal logic}\label{sec:CoalgModalLog}
\begin{wrapfigure}[8]{r}{0.4\textwidth}
\vspace{-1.9cm}
\begin{diagram}[width=3.5em,height=2.15em]
  \elements{\Psi} & \lTo^{\bar \tfunctor} & \op{\cat A}\\
   \dTo^p & & \dEMap\\
  \cat C & \pile{\rTo~\sfunctor \\ {\scriptscriptstyle\perp} \\ \lTo~\tfunctor} & \op{\cat A} & \eqnum\label{eq:KupkeRot_Setup} \\
  \dTo^F & \Downarrow \delta & \dTo_{L}\\
  \cat C & \pile{\rTo~\sfunctor \\ {\scriptscriptstyle\perp} \\ \lTo~\tfunctor} & \op{\cat A}
\end{diagram}
\end{wrapfigure}
The `partial' characterisation of behavioural equivalence as a fibred notion (cf. Theorems~\ref{thm:behequiv->fibredcoalg} and \ref{thm:fibredcoalg->behequiv}) enables us to use the dual  adjunction framework 
of Kupke and Rot \cite{KupkeRot_ExprLog_2020} 
in \eqref{eq:KupkeRot_Setup} to develop a logical characterisation of behavioural equivalence. It should be noted that, although this framework can handle behavioural preorders and distances, we prove our results only for behavioural equivalence, i.e. in the context of Assumptions~\ref{assum:product} and \ref{assum:presEq}. Below we explain the role of various functors drawn in \eqref{eq:KupkeRot_Setup} in an incremental manner; subsequently, we will establish our general adequacy and expressivity results (Theorems~\ref{thm:fibredbehequiv->logequiv} and \ref{thm:KleisliQu}) for behavioural equivalences. 


The fibration $\elements{\Psi} \rTo^p \cat C$ will be used to define (internally) a behavioural equivalence of interest. Often it is defined as a colimit of a diagram resembling the final sequence in a fibre (cf.\thinspace \cite{hasuo_kataoka_cho_2018}). More abstractly, we assume that the indexed category $\Psi^F_{\lambda}$ (recall this notation from Proposition~\ref{prop:technicalresult}) has indexed final objects, for some indexed morphism $\Psi \rTo^\lambda \Psi F$.

\begin{restatable}{lemma}{IndexFinal}\label{lemma:indexfinal}
  Suppose $\op{\cat C} \rTo^\Phi \ccat$ has indexed final objects (i.e., the final object exists in each fibre $\Phi X$) and the reindexing functor $f^*$ preserves these final objects. Then there is a functor $\cat C \rTo^{\1} \elements{\Phi}$ that is right adjoint to $p$.
\end{restatable}

\begin{wrapfigure}[4]{r}{0.45\textwidth}
\vspace{-1.6cm}
\begin{diagram}[width=2.5em,height=2.5em]
    \coalg{F_\lambda}{\elements{\Psi}} = \elements{\Psi_\lambda^F} & \rTo & \elements{\Psi}\\
    \dTo^{p^F_\lambda} \dashv \uTo_{\1^\lambda} &\eqnum\label{eq:largestbehequiv}& \dTo_p\\
    \coalg{F}{\cat C} &\rTo & \cat C
  \end{diagram}
\end{wrapfigure}
Usually, $\1$ is used (called the \emph{truth functor} \cite{jacobs-fibrations} in the context of logic) when the underlying fibration $\elements{\Phi} \rTo^p \cat C$ has indexed final objects. However when $\Psi^F_\lambda$ satisfies the conditions of the previous lemma, it results in a functor $\coalg{F}{\cat C} \rTo^{\1^\lambda} \coalg{F_{\lambda}}{\elements{\Psi}}$ (which we call \emph{the behavioural conformance} functor) that maps a coalgebra $X \rTo^\alpha FX$ to the terminal element in $\coalg{\alpha^*\circ\lambda_X}{\Psi X}$ denoted as $1^\lambda_X$. Note that $1^\lambda_X$ 
in our applications will correspond to the largest behavioural equivalence on a given system. Moreover, it is not hard to arrive at the adjoint situation as indicated in \eqref{eq:largestbehequiv}. So, in other words, the behavioural conformance functor is right adjoint to the forgetful functor that witnesses the fibration of behavioural conformance 
on coalgebras.

\begin{example}\label{ex:LTS}
  Consider the indexed category $\Psi$ induced by binary relations on sets and a labelled transition system modelled as a coalgebra $X \rTo^\alpha (\power X)^\act$, i.e., our $\cat C=\set,F=(\power\_)^\act$. Consider the function $\Psi X \rTo^{\lambda_X} \Psi F X$ that maps a relation $R\in\Psi X$ to a relation $\lambda_X R\in \Psi FX$
  (below $q,q'\in(\power X)^\act$):
\[  q \mathrel{\lambda_X R} q' \iff \forall_{a,x}\exists_{x'}\ (x \in qa \implies x'\in q'a) \land \forall_{a,x'}\exists_{x}\ (x' \in q'a \implies x\in qa).
\]
It is well known (as first noted in \cite{Hermida_Jacobs:str_coind_fibration}) that a bisimulation relation is a $\alpha^*\circ \lambda$-coalgebra. Moreover, bisimilarity $\bisim_X$ (the largest bisimulation relation on $X$) corresponds to the final object in $\Psi_\lambda (X,\alpha)$, i.e., $\1^\lambda(X,\alpha)=\ (X,\alpha,\bisim_X)$.
\end{example}

As for the dual adjunction $\sfunctor\dashv \tfunctor$ in
\eqref{eq:KupkeRot_Setup}, it provides a connection (cf.\thinspace
\cite{DAdj:origin}) between states and theories (the formulae
satisfied by a state).
The syntax of the logic is given by a functor $\cat A \rTo^L \cat A$ and it is assumed that the initial algebra $L\mathcal A \rTo^h \mathcal A \in\cat A$ exists for $L$, which models the typical Lindenbaum algebra induced by the term algebra. Lastly, the natural transformation $\delta$ gives the one-step interpretation to the formulae which can be 
given its mate $\theta$ as described below (cf.\thinspace \cite[Proposition~2]{JacobsSokolova_expressivity_modal_Logic}).
\begin{proposition}\label{prop:mate}
 Given $\cat C \rTo^F \cat C, \cat A \rTo^L \cat A,\cat C \pile{\rTo^\sfunctor \\ \lTo_\tfunctor} \op{\cat A}$ with $\sfunctor\dashv \tfunctor$, there is a correspondence between  $F\tfunctor \rTo^{\delta} \tfunctor L$ and $\sfunctor F \rTo^{\theta} L\sfunctor$. 
\end{proposition}
\noindent
Now given a coalgebra $X \rTo^\alpha FX \in\cat C$, the semantics $\mathcal A \rTo^{\semantics{\_}_X} \sfunctor X$ of the logic $(L,\delta)$ is given by the universal property of the initial algebra $L\mathcal A \rTo^h \mathcal A$. In particular, it is the unique arrow in $\cat A$ that makes the following diagram (drawn on the left) commutative. And the transpose of the semantics map $\semantics{\_}_X$ under $\sfunctor\dashv \tfunctor$ gives a `theory' map $X \rTo^{\theory{\_}_X} Q\mathcal A$; it is the unique arrow in $\cat C$ that makes the following diagram on the right commutative.
\vspace{-0.2cm}
      \begin{diagram}[width=3em,height=2em]
        L\mathcal A & \rTo^{L \semantics{\_}_X} & L\sfunctor X & \rTo^{\theta_X} & \sfunctor FX && X && \rTo~{\theory{\_}_X} && \tfunctor \mathcal A\\
        \dTo^{h} & & \eqnum\label{eq:semantics} & & \dTo_{\sfunctor\alpha} && \dTo^{\alpha} && && \dTo_{\tfunctor h}\\
        \mathcal A & & \rTo~{\semantics{\_}_X} & & \sfunctor X && FX & \rTo^{F\theory{\_}_X} & F\tfunctor \mathcal A & \rTo^{\delta_{\mathcal A}} & \tfunctor L \mathcal A
      \end{diagram}
\noindent
Once these niceties are set up, one can argue when a logic ($L,\delta$) is adequate and expressive. Intuitively, a logic ($L,\delta$) is adequate if behaviourally equivalent states satisfy the same logical formulae; while an adequate logic is expressive if logically equivalent states are also behaviourally equivalent. The formulation below is a straightforward formulation of adequacy and expressivity given in \cite{KupkeRot_ExprLog_2020} using the language of indexed categories.

\begin{definition}
  Suppose the behavioural conformance functor $\1^\lambda$ exists (for some $\lambda$) with $\op{\cat A} \rTo^{\bar \tfunctor} \elements{\Psi}$ such that $p \circ \bar \tfunctor = \tfunctor$. Then a logic $(L,\delta)$ is \emph{adequate} (resp.\thinspace \emph{expressive}) w.r.t.\thinspace $\bar\tfunctor$ if $(X,1_X^\lambda) \rTo^{\theory{\_}_X} \bar \tfunctor \mathcal A$ is a (resp.\thinspace Cartesian) map in $\elements{\Psi}$, for every coalgebra $X \rTo^\alpha FX \in \cat C$.
\end{definition}

The role of $\bar\tfunctor$ is to encode a relationship between the theories of any two states (cf.\thinspace \cite{KupkeRot_ExprLog_2020}); so we let $\bar\tfunctor = \Eq \circ \tfunctor$ in the context of behavioural equivalence. Next we state the main result of this section, which is a refinement of adequacy and expressivity results given in \cite{KupkeRot_ExprLog_2020}.
\begin{restatable}{theorem}{FibBehEquivImpLogEquiv}\label{thm:fibredbehequiv->logequiv}
Under the assumptions of Theorem~\ref{thm:behequiv->fibredcoalg}, if $\bar\tfunctor$ has a left adjoint $\bar\sfunctor$, the logic $(L,\delta)$ is adequate.
Moreover it is expressive if $\forget{\delta_{\mathcal A}}$ is injective.
\end{restatable}
\noindent
In short, the Kupke-Rot logical setup for behavioural equivalence can be summarised as drawn left on the next page. Now if our indexed category $\Psi$ satisfies \ref{assum:quotient} (like in the case of coKleisli and Eilenberg-Moore categories), then $\bar \sfunctor=\sfunctor\circ \Qu$ as indicated in \eqref{eq:KupkeRot_behequivSetup}. However, in the case of Kleisli categories we will construct $\bar\sfunctor$ under some restrictions (cf.\thinspace Theorem~\ref{thm:KleisliQu}).
\vspace{-0.3cm}
\begin{equation}\label{eq:KupkeRot_behequivSetup}
  \begin{diagram}[width=4em,height=2.5em]
  \elements{\Psi} & \pile{\rTo^{\bar\sfunctor} \\ {\scriptscriptstyle \perp} \\ \lTo_{\Eq \circ\tfunctor} } & \op{\cat A} && \elements{\Psi} & \pile{\rTo^{\sfunctor\circ \Qu} \\ {\scriptscriptstyle \perp} \\ \lTo_{\Eq \circ\tfunctor} } & \op{\cat A}\\
  \uTo^{\Eq} \dTo_p & & \dEMap && \dTo^{\Qu} \dashv \uTo~{\Eq} \phantom{\dashv}\dTo_p & & \dEMap\\
  \cat C & \pile{\rTo^{\sfunctor} \\ {\scriptscriptstyle \perp} \\ \lTo_{\tfunctor} } & \op{\cat A} && \cat C & \pile{\rTo^{\sfunctor} \\ {\scriptscriptstyle \perp} \\ \lTo_{\tfunctor} } & \op{\cat A}
\end{diagram}
\end{equation}
\vspace{-0.7cm}
\subsection*{Construction of $\bar \sfunctor$ for Kleisli categories}
Unfortunately, arbitrary (co)limits in general do not exists in a Kleisli category. For instance, one of our working categories $\klcat{\power}\cong\setrel$ (the category of sets and relations) does not have all coequalisers, but $\setrel$ has a reflective subcategory $\op{\set}$ that does. The presence of these coequalisers in the reflective subcategory will then be used to construct $\bar\sfunctor$.
\begin{definition}
  A subcategory $\cat B \rMono^{\jmath} \cat C$ is \emph{reflective} when the inclusion functor $\jmath$ has a left adjoint $\reflector$ (often called as reflector).
\end{definition}
\begin{theorem}[\cite{ABHKMS12}]\label{thm:reflection_coalg}
  If $\cat B \rMono^\jmath \cat C$ is a reflective subcategory of $\cat C$ and $\cat C \rTo^F \cat C$ preserves $\cat B$, i.e.,
  $\forall_{B,f\in \cat B}\ (FB \in \cat B \land Ff\in \cat B)$ and $F \circ \jmath = \jmath \circ F,$
  then 
  $
  \coalg {F}{B} \pile{\lTo~{\bar\reflector} \\ \rTo~{\bar \jmath} } \coalg {F}{C}
  $
  with $\bar\reflector \dashv \bar\jmath $.
  Here, $\bar\jmath$ is the obvious inclusion.
\end{theorem}
The reflector $\bar\reflector$ typically results in 
a form of (on-the-fly) determinisation (cf. Example~\ref{ex:NDAQuotient}).
Moreover, in our case studies, these reflective subcategories will also take the place of algebras in \eqref{eq:KupkeRot_behequivSetup}, and if these reflective subcategories have coequalisers, then we can construct $\bar \sfunctor$ in general.

So let $\cat B=\op{\cat A},\sfunctor=\reflector,\tfunctor=\jmath$, and $(X,R) \in \elements{\Psi}$. Then the idea is to use the following series of transformations (depicted below on the left) to construct $\bar\sfunctor$ as the equaliser of the parallel arrows $\sfunctor p_1',\sfunctor p_2'\in\cat A$. 
Below $p_i$ (for $i\in\{1,2\}$) are the obvious projection functions and each $p_i'$ is the transpose of $p_i$ under the free-forgetful adjunction $\iota\dashv\forget{\_}$.

\vspace{-0.35cm}
\begin{minipage}[t]{0.4\textwidth}
  \[
\infer{
\sfunctor \iota R \parrows{\sfunctor p_1'}{\sfunctor p_2'} \sfunctor X \in \op{\cat A}
}{
\infer{
\iota R \parrows{p_1'}{p_2'} X \in \klcat{T}
}{R \parrows{p_1}{p_2} \forget X\in \set}
}
\]
\end{minipage}
\begin{minipage}[t]{0.5\textwidth}
\vspace{-0.4cm}
  \begin{diagram}[width=3.5em,height=2.65em]
  \bar\sfunctor (X,R) & \rTo^e & \sfunctor X  & \parrows{\sfunctor p_1'}{\sfunctor p_2'}  & \sfunctor (\iota R)\\
  \uExistsMap_{\bar\sfunctor f}&& \uTo^{\sfunctor f}& \eqnum\label{eq:coequaliser}& \uTo^{\sfunctor \iota g} \\
  \bar\sfunctor(Y,S) &\rTo_{e'}& \sfunctor Y  & \parrows{\sfunctor q_1'}{\sfunctor q_2'} & \sfunctor (\iota S)
  \end{diagram}
\end{minipage}

\noindent
Let $(\bar\sfunctor(X,R),e)$ be the equaliser of $\reflector p_i'$ in $\cat A$. Now $(X,R) \rTo^f (Y,S)\in\elements{\Psi}$ means that $X\rTo^f Y\in\klcat T$ and $R \subseteq \inverse{(\forget f \times \forget f)} S$. So there is a function $R \rTo^g S$ such that $\forget{f} \circ p_i = q_i \circ g$ with $p_i,q_i$ being the obvious projections when the relations $R,S$ are viewed as spans in $\set$. Moreover $f\circ p_i' = q_i' \circ \iota g$ due to the naturality of the counit of $\iota\dashv\forget{\_}$.
So the two squares in \eqref{eq:coequaliser} commute and the universal property of equalisers gives the unique $\bar\sfunctor f$.
\begin{restatable}{theorem}{KleisliQuotients}\label{thm:KleisliQu}
  Under \ref{assum:product} and $\op{\cat A}$ being a reflective subcategory of $\klcat{T}$ having coequalisers, the above defined $\bar\sfunctor$ is a functor and left adjoint to $\bar \tfunctor=\Eq \circ \tfunctor$.
\end{restatable}

\section{Lifting of predicate and relation liftings}\label{sec:indMorKleisli}
In this paper, the indexed categories corresponding to predicates (or relations) on our working category (like (co)Kleisli or Eilenberg-Moore categories) are always induced by lifting an indexed category on the underlying base category $\set$ (for instance, recall $\Psi$ from \eqref{eq:indexedcatRelations}). In a similar spirit, our aim is to construct indexed morphisms on our working category by lifting an indexed morphism on $\set$. So consider an indexed category $\Phi$ of predicates given by $\Phi=\contrapower\circ\forget{\_}$ and an endofunctor $\cat C \rTo^{\bar F} \cat C$ modelling the branching type of behaviour of interest. 

%

\paragraph*{Lifting of predicate liftings}
Next we give a recipe to construct a predicate lifting, i.e., an indexed morphism of type $\Phi \rTo^\lambda \Phi \bar F$. In particular, we need an endofunctor $\set \rTo^G \set$, a predicate lifting $\contrapower \rTo^{\sigma} \contrapower G$, and a natural transformation $\gamma$ as indicated below on the left.

\vspace{-0.5cm}
\begin{minipage}[t]{0.4\textwidth}
  \begin{diagram}[width=5em,height=2.5em]
  \cat C & \rTo_{\bar F} & \cat C\\
  \dTo^{\forget{\_}} & \Downarrow\gamma & \dTo_{\forget{\_}}\\
  \set & \rTo^G & \set
\end{diagram}
\end{minipage}
\begin{minipage}[t]{0.6\textwidth}
  \begin{diagram}[nohug,width=5em,height=1.25em]
    TFTX & \rTo_{\gamma_{TX}} & GTTX\\
    \dTo^{T\vartheta_X } & & \\
    TTFX & \eqnum\label{eq:thetamu_compatible}& \dTo_{G\mu_X} \\
    \dTo^{\mu_{FX}} &&\\
    TFX & \rTo^{\gamma_X }& GTX
  \end{diagram}
\end{minipage}

\vspace{0.3cm}
\noindent
As a result, we can define $\lambda$ by the composition:
\begin{equation}\label{eq:liftOfPlift}
  \contrapower \forget{X} \rTo^{\sigma_{\forget{X}}} \contrapower G\forget{X} \rTo^{\gamma_X^*} \contrapower\forget{F X}.
\end{equation}
\begin{theorem}
  \label{th:lambda-indexed-morphism}
  The above mapping $\lambda$ is an indexed morphism.
\end{theorem}
Note that 
in the case of coKleisli  and Eilenberg-Moore categories, we simply let $G=F$ and $\bar F$ be a coKleisli extension/Eilenberg-Moore lifting of $F$, which results in a distributive law of type $\forget{\bar F\_} \rTo^\gamma G \forget{\_}$; in the case of Eilenberg-Moore categories, such natural transformations are also known as EM-laws \cite{Jacobs_etal:trace_EM_KL}. 

\noindent
In the case of Kleisli categories, the situation is slightly
complicated. This stems from the fact that the distributive law $FT \rTo^\vartheta
TF$ (which induces a Kleisli extension $\bar F$ of $F$) results in a
natural transformation in the `wrong' direction $F \forget{\_}
\rTo^\vartheta \forget{\bar F\_}$. However, in various
applications, $G$ is typically associated with the branching type of a
deterministic version of the corresponding system of interest (such as $G=\_^\act \times 2$ in the case of NDA) exists. The next result helps in finding such a distributive law $\gamma$ for a given $G$ in a more elementary way.
\begin{restatable}{lemma}{thetamuComp}\label{lemma:theta&mu_compatible}
  Let $\bar F$ be a Kleisli extension of $F$ induced by a distributive law $FT \rTo^\vartheta TF$. Then a natural transformation $TF \rTo^\gamma GT$ compatible with
  $\vartheta$ and $\mu$ (i.e., Square~\ref{eq:thetamu_compatible}
  commutes) induces a distributive law $\forget{\_}\circ \bar F \rTo G\circ \forget{\_}$.   Moreover, the converse also holds.
\end{restatable}
\begin{remark}
  Note that the compatibility property in the above lemma appeared in \cite{Jacobs_etal:trace_EM_KL} as part of an `extension' natural transformation. In short, the properties of an extension natural transformation are more stringent than of Lemma~\ref{lemma:theta&mu_compatible}. 
\end{remark}
\paragraph*{Lifting of relation liftings}
The above idea can also be used to construct a relation lifting, i.e., an indexed morphism of type $\Psi \rTo \Psi \bar F$, where $\Psi$ is the indexed category of abstract relations given in \eqref{eq:indexedcatRelations}. So now given a relation lifting $\contrapower (X\times X) \rTo^{\sigma_X} \contrapower (GX\times GX)$ of $G$, then we can define $\Psi \rTo^{\bar\lambda} \Psi\bar F$:
\vspace{-0.2cm}
\begin{diagram}[height=3em]
  \Psi X &\rTo^{\langle \pi_1^X,\pi_2^X\rangle_!}& \contrapower (\forget X \times \forget X) & \rTo^{\sigma_{\forget{X}}} & \contrapower (G\forget{X} \times G\forget{X}) \\
  & & & & \dTo_{\inverse{(\gamma_X\times\gamma_X)}} && \eqnum\label{eq:liftofRlift}\\ &&&& \contrapower{(\forget{\bar FX}\times\forget{\bar FX})} & \rTo^{\inverse{\langle \pi_1^{FX},\pi_2^{FX}\rangle}} & \Psi\bar FX.
\end{diagram}
Here, we use the fact that $\elements{\contrapower}$ is a bifibration \cite{jacobs-fibrations}, i.e., for any function $X \rTo^f Y$ the reindexing functor $\inverse f$ has the direct image functor $f_!$ as its left adjoint. Now under the following assumption we can show that the $\bar\lambda$ is indeed an indexed morphism.
\begin{enumerate}[label=\textbf{A\arabic*},start=4]
\item\label{assum:projpullback} The square drawn in \eqref{eq:squarepbk} is a weak pullback in $\set$ for every $f\in\cat C$.
\end{enumerate}
The above assumption ensures that, in the context of $\elements{\contrapower}$, the square in \eqref{eq:squarepbk} satisfies the Beck-Chevalley condition, i.e., the following equation holds
\[
\langle\forget{\pi_1^X},\forget{\pi_2^X}\rangle_! \circ \inverse{(f\otimes f)} = \inverse{(\forget{f} \times \forget f)}\circ \langle\forget{\pi_1^Y},\forget{\pi_2^Y}\rangle_! \enspace.
\]
In turn this equation is used in diagram chasing to show that $\bar\lambda$ is a natural transformation. Furthermore, \ref{assum:projpullback} trivially holds when $\cat C$ is a Kleisli or Eilenberg-Moore category (cf.\thinspace Corollary~\ref{cor:projpullback}) because the canonical function $\langle\forget{\pi_1^X},\forget{\pi_2^X}\rangle$ is a bijection for each $X\in\cat C$. In other words, when $\cat C$ is a Kleisli or Eilenberg-Moore category, $\Psi$ can be alternatively defined as the composition $\contrapower\circ (\forget{\_} \times \forget{\_})$.
\vspace{-0.5cm}
\begin{restatable}{theorem}{RelLiftIndMorphism}\label{thm:RelLiftIndMorphism}
Under \ref{assum:projpullback} $\bar\lambda$ as defined in \eqref{eq:liftofRlift} is an indexed morphism.
\end{restatable}
\begin{corollary}\label{cor:projpullback}
  If the forgetful functor $\cat C \rTo^{\forget{\_}} \set$ is product preserving, then \ref{assum:projpullback} is always satisfied. As a result, $\bar\lambda$ defined in \eqref{eq:liftofRlift} is an indexed morphism.
\end{corollary}

\section{Nondeterministic automata (NDA)}\label{sec:KleisliEx:NDA}
Recall the necessary parameters from Example~\ref{ex:nda} for coalgebraic modelling of an NDA, i.e., $T=\power$, $\cat C=\klcat{\power}\cong\setrel$, $F=\nda{\_}$, the distributive law $\vartheta$ given in \eqref{eq:thetaNDA}, and the free-forgetful adjunction $\iota\dashv\forget{\_}$ associated with any Kleisli category. To apply Theorem~\ref{thm:KleisliQu}, we also recall the reflective subcategory $\op{\set}$ of $\setrel$ from \cite{ABHKMS12}. 
Below $X \rTo^f Y \in \set$ and $X \rTo^g Y \in\setrel$:
{\allowdisplaybreaks
\begin{align*}
\jmath X =& X & Y\rTo^{\jmath f} X \in \setrel && y \mathrel {\jmath f} x \iff fx =y \\
\reflector X =& \power X & \reflector X \rTo^{\reflector g} \reflector Y\in\op\set && \reflector g(V)=\{x\mid \exists_{y\in V}\ x \mathrel g y\}.
\end{align*}
}

\begin{wrapfigure}[2]{r}{0.4\textwidth}
\vspace{-1cm}
  \begin{tikzpicture}
\node [draw, circle] (S1) at (-2,0) {$x$}; 			
\node [accepting,draw, circle](S3) at (0,0) {$z$};

\draw  [->,thick] (S1) to node [above]{$a$} (S3);
\node [draw, circle] (T2) at (2,0) {$y$}; 			
\draw  [->,thick] (T2) to node [above]{$a,b$} (S3);
\end{tikzpicture}
\end{wrapfigure}
\noindent
Next we illustrate the definition of $\bar \sfunctor$ and how the unit of $\bar\sfunctor\dashv\bar\tfunctor$ maps an NDA to the largest subautomaton (respecting language equivalence) obtained after backward determinisation of the given NDA. It is worthwhile to note that the abstract and concrete state space coincide (up to isomorphism) in the case of NDA because forgetful functor preserves products, i.e., $\power (X+X) =\forget{X+X} \cong \forget X \times \forget X = \power X \times\power X$. Therefore, as mentioned earlier, we will simplify our presentation by working with the indexed category $\contrapower\circ (\forget{\_} \times \forget{\_})$.

\begin{example}\label{ex:NDAQuotient}
Consider the NDA drawn above on the right with the accepting state $z$ as a coalgebra $X\rTo^\alpha\act\times X +1 \in \klcat{\power}$. Logical equivalence $\simeq$ is the least
  equivalence that equates $\{x,y\}$ with $\{y\}$ (both accept
  the language $\{a,b\}$) and $\{x,y,z\}$ with $\{y,z\}$ (both accept
  the language $\{a,b,\epsilon\}$). Also, for any $U,U'\subseteq X$ such that $U \mathrel R U'$ we have $(U,U') \mathrel {p_1'} x \iff x\in U$ and $(U,U') \mathrel {p_2'} x \iff x\in U'$. Note that $p_i'$ are transpose of $p_i$ for $i\in\{1,2\}$ (see \eqref{eq:coequaliser}). So the equaliser $\bar\sfunctor(X,\simeq)$ of $\reflector p_i'$ can be computed as follows:
  \[
  \bar\sfunctor(X,\simeq)  = 
  \big\{W\in \power X \mid \forall_{U,V}\ (U \cap W \neq \emptyset \land U \simeq V) \implies V \cap W \neq \emptyset\big\}.
  \]
The arrow $\bar\sfunctor(X,\simeq) \rTo^{\beta} F\bar\sfunctor(X,\simeq)\in\op\set$ is defined by the following (depicted on the left) universal property of equaliser in $\set$. Here, $\bar \reflector \alpha$ is the backward determinisation of the
given coalgebra (as described, e.g. in \cite{ABHKMS12} as a deterministic automaton accepting the reverse language), i.e. it maps $(a,U)\mapsto \{x\mid
\exists_{x'\in U} (a,x')\in\alpha(x)\}$ and $\bullet \mapsto \{x
\mid \bullet \in \alpha(x)\}$. Thus, in essence, $\beta$ acts like
$\bar\reflector\alpha$ on the elements of $\bar\sfunctor(X,\simeq)$.

\vspace{-0.2cm}
\hspace{-0.67cm}
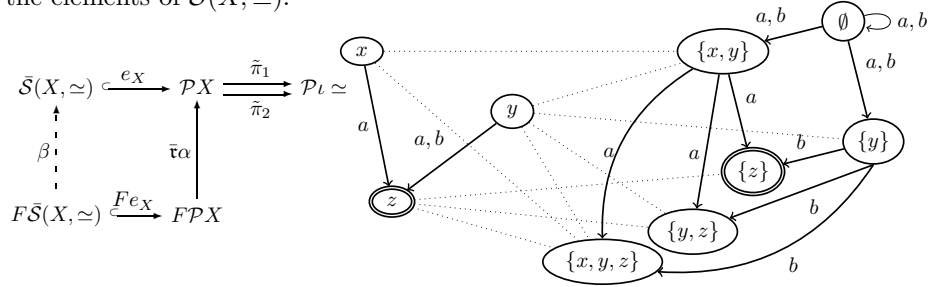
\begin{minipage}{0.3\textwidth}
\scalebox{0.8}{
  \begin{diagram}[width=3em]
  \bar\sfunctor(X,\simeq) & \rMono^{e_X} & \power X & \parrows{\tilde \pi_1 }{\tilde \pi_2} & \power \iota \simeq\\
  \uExistsMap^{\beta} & & \uTo^{\bar\reflector\alpha}\\
   F\bar\sfunctor(X,\simeq)& \rMono^{Fe_X} &  F\power X
\end{diagram}}
\end{minipage}\hspace{0.6cm}
\begin{minipage}{0.6\textwidth}
  \scalebox{0.8}{
\begin{tikzpicture}
\tikzstyle{state}=[thick, draw, ellipse, minimum width=20pt, inner sep=2.5pt,
    align=center]
\node[state] (x1) at (-2.5,3) {$\emptyset$};
\node[state] (x2) at (-4.5,2.5) {$\{x,y\}$};
\node[state] (x3) at (-2,1) {$\{y\}$};
\node[state,accepting] (x4) at (-4,0.5) {$\{z\}$};
\node[state] (x5) at (-5,-0.5) {$\{y,z\}$};
\node[state] (x6) at (-6.5,-1) {$\{x,y,z\}$};
\node[state] (lx1) at (-8,1.5) {$y$};
\node[state] (lx2) at (-10.5,2.5) {$x$};
\node[state,accepting] (lx3) at (-10,0) {$z$};
\path[->, thick]
 (x1) edge node[above left] {$a,b$} (x2)
 (x1) edge node[above right] {$a,b$} (x3)
 (x2) edge node[above right] {$a$} (x4)
 (x3) edge node[above left] {$b$} (x4)
 (x2) edge node[below left] {$a$} (x5)
 (x3.south) edge node[below right] {$b$} (x5)
 (x2) edge[bend right] node[below left] {$a$} (x6)
 (x3.south) edge[bend left] node[below right] {$b$} (x6);
 \draw  [->] (x1) edge [loop right] node {$a,b$} (x1);
\draw[dotted]
	(lx2) edge (x2)
      (lx2) edge (x6)
      (lx1) edge (x2)
      (lx1) edge (x3)
      (lx1) edge (x5)
      (lx1) edge (x6)
      (lx3) edge (x4)
      (lx3) edge (x5)
      (lx3) edge (x6);
\draw[->,thick]
(lx2) edge node[left] {$a$} (lx3)
(lx1) edge node[above left] {$a,b$} (lx3);
\end{tikzpicture}}
\end{minipage}

\medskip
\noindent
In this example, we obtain as $\beta$ the automaton drawn above on
the right with six states. The relation $\kappa_X$ indicated by dotted line is the transpose of $e_X$ w.r.t.\thinspace $\reflector\dashv \jmath$;
concretely, $x \mathrel{\kappa_X} U \iff x\in U$. Furthermore $\kappa_X\in\setrel$ is a witnessing coalgebra homomorphism because
$\bar F(\kappa_X) \circ \alpha = \beta \circ
\kappa_X$.  
  Note that $|\kappa_X|$ maps both $\{x,y\},\{y\}$ to
  $\{\{x,y\},\{y\},\{y,z\},\{x,y,z\}\}$, witnessing the fact that they
  are language equivalent.
  Hence, the
  coequaliser gives us the largest sub-automaton of the
  backwards determinisation that respects logical
  equivalence (removing $\emptyset$, $\{y,z\}$, and $\{x,y,z\}$ will result in the smallest such sub-automaton).
\end{example}

\paragraph*{Predicate liftings for NDAs}
To apply techniques from Section~\ref{sec:indMorKleisli}, we set $G=\_^\act \times 2$ and define $\gamma$ as follows:
\begin{align*}
   &\power(\act\times X + 1) \rTo^{\gamma_X} (\power X)^\act \times 2 \qquad\quad\quad \bar U \mapsto (\gamma_X^\act \bar U,\gamma_X^2\bar U), \\
   \text{where} &\ \gamma_X^\act\bar U(a) = \{x \mid (a,x)\in\bar U\} \qquad \text{and}\qquad
   \gamma_X^2\bar U = 1 \iff \bullet \in \bar U.
\end{align*}
Moreover, from \cite{Jacobs_etal:trace_EM_KL} we know that $\gamma$ is compatible with $\theta$ and $\bigcup$ in the sense of
Lemma~\ref{lemma:theta&mu_compatible}. 
Next consider the family of liftings $\contrapower X \rTo^{\sigma_X^a} \contrapower(X^\act\times 2)$ (for each $a\in\act$) and $\contrapower X \rTo^{\sigma_X^\term} \contrapower (X^\act\times 2)$:
$
U\mapsto \{(p,b)\in X^\act\times 2 \mid p(a) \in U\}$
and 
$U \mapsto \{(p,1) \mid p\in X^\act\}$, respectively.
\begin{restatable}{lemma}{predliftnda}\label{lem:predlift_nda}
  The above mappings $\sigma_X^a$ and $\sigma_X^\term$ are indexed morphisms.
\end{restatable}
\noindent
Thanks to Theorem~\ref{th:lambda-indexed-morphism},
$\lambda^a=\inverse{\gamma}\circ \sigma^a,\lambda^\term=\inverse{\gamma}\circ \sigma^\term$ are valid predicate liftings for the functor $\setrel \rTo^{\overline {\act\times\_+1}} \setrel$. Moreover, for any $\mathbb U \subseteq\power X$ we find:
{\allowdisplaybreaks
  \begin{align*}
  \lambda_X^a&(\mathbb U) =\ \inverse{\gamma_X}\sigma_{\power X}^a(\mathbb U) & \lambda_X^\term&(\mathbb U) =\ \inverse{\gamma_X}\sigma_{\power X}^\term(\mathbb U)\\
  =&\ \inverse{\gamma_X}\{(p,b) \in (\power X)^\act \times 2 \mid p(a) \in \mathbb U\} & =&\ \inverse{\gamma_X}\{(p,1) \mid p\in (\power X)^\act \}\\
  =&\ \left\{\bar U \in\act\times X + 1 \mid \{x \mid (a,x)\in\bar U\} \in \mathbb U \right\} & =&\ \left\{\bar U \in\act\times X + 1 \mid \bullet \in \bar U\right\}.
\end{align*}
}
The above indexed morphisms $\lambda^a,\lambda^\term$ induce modalities that can be interpreted on determinised automata.
Given an NDA $(X,\act,\rightarrow,\downarrow)$ with $\downarrow\subseteq X$ being the termination predicate (or $X \rTo^\alpha \power(\act\times X + 1)\in\set$), then the determinised automaton has state space $\power X$ with dynamics given by the SOS rules or abstractly by the composition $\gamma_X \circ \bigcup \circ \power\alpha$:
\[\infer {U \step a U_\alpha} {U\subseteq X \quad U_\alpha=\{x'\mid \exists_{x\in U}\ x \step a x'\}}
\qquad \infer {U \term}{U\subseteq X \qquad \exists_{x\in U}\ x \term}\enspace.\]
In turn, we can now rewrite the two modalities to a simpler form:
{\allowdisplaybreaks
\begin{align*}
  \inverse{\forget{\alpha}}&\lambda_X^a \mathbb U =\  \inverse{\forget\alpha} \{\bar U \mid \{x \mid (a,x) \in \bar U\} \in \mathbb U\} & \inverse{\forget{\alpha}}&\lambda_X^\term \mathbb U =\  \inverse{\forget\alpha} \left\{\bar U \mid \bullet \in \bar U\right\}\\
  =&\ \{U\mid U \step a U_\alpha \implies U_\alpha\in\mathbb U\} & =&\ \{U \mid U\term\}.
\end{align*}
}
\vspace{-0.7cm}
\paragraph*{Language equivalence through relation lifting}
First note that products exists in $\klcat{\power}$ and are given by disjoint union. Moreover, Assumption~\ref{assum:projpullback} is trivially satisfied since the forgetful functor preserves all limits (cf.\thinspace Corollary~\ref{cor:projpullback}). So we can create a relation lifting of $\bar F$ from the following relation lifting $\contrapower (X \times X) \rTo^{\bar\sigma_X} \contrapower (GX \times GX)$ of $G$ (below $R\subseteq X \times X$):
\[   (p,b) \mathrel{\bar\sigma_X R}  (p',b') \iff b=b' \land \forall_{a\in\act}\ pa \mathrel{R} p'a \enspace.\]
\vspace{-0.6cm}
\begin{restatable}{lemma}{GRelLift}\label{lem:GRelLift}
  The mapping $\bar\sigma$ defined above is an indexed morphism.
\end{restatable}
\noindent
So $\Psi \rTo^{\bar \lambda} \Psi \bar F$ given in \eqref{eq:liftofRlift} is an indexed morphism.  Concretely, it maps a relation $R\subseteq \power X \times\power X$ to a relation $\bar\lambda_X R \subseteq \power FX \times \power FX$: $\bar U \mathrel{\bar\lambda_XR} \bar U'$ iff
\[
\left(\bullet\in\bar U \iff \bullet \in \bar U'\right)  \land \left( \forall_{a\in\act}\ \{x \mid (a,x)\in\bar U\} \mathrel R \{x \mid (a,x) \in \bar U'\}\right) \enspace.
\]
\vspace{-0.6cm}
\begin{restatable}{lemma}{LambdaCofilter}\label{lem:lambda_cofilter}
  The indexed morphism $\bar\sigma$ preserves arbitrary intersections at each component; therefore, so does the predicate lifting $\bar\lambda$. Moreover, $\bar\lambda$ satisfies \ref{assum:presEq}.
\end{restatable}
\vspace{-0.4cm}
\begin{restatable}{theorem}{NDALangEquivFibred}\label{thm:NDALangEquivFibred}
  Let $X \rTo^\alpha FX\in\setrel$ be an NDA. Then language equivalence $\simeq_X\subseteq \power X \times \power X$ on the determinised system is a $\Psi(\alpha)\circ\bar\lambda_X$-coalgebra, i.e., $\simeq_X \subseteq \inverse{(\forget{\alpha}\times\forget{\alpha})} (\bar\lambda_X \simeq_X)$. Moreover,  $\simeq_X = \inverse{(\forget{f} \times \forget f)} \simeq_Y$ for any coalgebra homomorphism $f$ between $(X,\alpha)$ and $(Y,\beta)$; thus, there is a behavioural conformance functor $\coalg{\bar F}{\setrel} \rTo^{\1^{\bar\lambda}} \coalg{\bar F_\lambda}{\elements{\Psi}}$.
\end{restatable}
\paragraph*{Logical characterisation of language equivalence}
Recall the adjoint situation $\reflector \dashv \jmath$ that witnesses $\op\set$ is a reflective subcategory of $\setrel$. We use this dual adjunction to model our logic because (intuitively) conjunction is not needed to characterise language equivalence. 
Thus we fix $\cat A=\set$, $S=\reflector$, and $T=\jmath$. Moreover, a left adjoint $\bar\sfunctor$ of $\bar\tfunctor$ exists due to Theorem~\ref{thm:KleisliQu}.

Since to establish language equivalence one needs to ascertain whether a word in $\fseq\act$ is accepting or not, so we take our syntax functor $L=\nda{\_}$. Note that the initial algebra of $L$ exists and is given by $\mathcal A = \fseq\act$. As for the one-step semantics given by a natural transformation $\delta$, we are going to define it (indirectly) by defining its mate $L\sfunctor X = \act\times \power X + 1 \rTo^{\theta_X} \power(\act\times X + 1) = \sfunctor \bar FX \in\set$
that 
acts on objects like the distributive law $\vartheta_X$ (see \eqref{eq:thetaNDA}). Note that, however, they differ in their naturality conditions.
\begin{restatable}{proposition}{ThetaNDA}\label{thm:thetatNDA}
  The above defined mapping $\theta$ is a natural transformation.
\end{restatable}
\noindent
The algebra $\act\times\fseq\act + 1 \rTo^{h} \fseq\act$ is given by the unary concatenation of words and the constant $\eseq$ (i.e., $h(a,w)=aw$ and $h\bullet=\eseq$). Consider the map $X \rTo^{\theory{\_}} \mathcal A \in \setrel$ that maps a state to the language accepted by it. 
\begin{restatable}{corollary}{NDALogChar}\label{cor:NDALogChar}
  The above map $\theory{\_}$ is indeed the theory map for a given NDA. So the logic $(L,\delta)$ is both adequate and expressive for language equivalence on determinised systems.
\end{restatable}
\section{Conditional transition systems: an application in coKleisli categories}\label{sec:CTS}
We next consider conditional transition systems (CTSs)
\cite{ABHKMS12}, strongly related to the featured transition systems
used for modelling software product lines
\cite{ccpshl:simulation-product-line-mc}. A conditional transition system is a compact
  representation of several transition systems -- one for each
  condition or product -- where transitions are labelled by products. Here we consider the simpler
case of conditional transition systems without upgrades and action
labels; their full treatment
\cite{BBKSW18_cts_coalgebraic,CTSwithUpgrades} is left for the future.
In our earlier work, CTSs were coalgebras living in the Kleisli category
induced by the reader monad $\_^\condset$ (for a fixed set of
conditions $\condset$). It is however more convenient to treat CTSs as coalgebras in coKleisli categories, hence we start with a 
relevant comonad $G = \condset\times\_$ whose counit is given by the projection of the second component and comultiplication $\Delta$ is given by the diagonal map. The map $\condset\times X \rTo^{\Delta_X} \condset \times \condset \times X$ is given by $(k,x) \mapsto (k,k,x)$ (for $k\in \condset,x\in X$).

Consider the coKleisli category $\coklcat G$ whose objects, just like in any Kleisli category, are sets; an arrow $X \rTo^f Y$ corresponds to a function $GX \rTo^f Y$.
Now there is a forgetful functor $\coklcat{G}\rTo^{\forget{\_}} \cat \set$; however, in contrast to the Kleisli setting, it is now left adjoint to the inclusion $\cat \set \rTo^{\iota} \coklcat{G}$. Concretely, this forgetful functor maps an object $X\mapsto GX$ and an arrow $X\rTo^f Y\in\coklcat{G}$ to a function $\forget f$ mapping $(k,x) \mapsto (k,f(k,x))$.

Next 
to model the branching type of CTSs,
take $\bar F$ to be the coKleisli extension \cite{BBKSW18_cts_coalgebraic} of $F=\power$ given by the following distributive law $\condset\times\power X \rTo^{\gamma_X} \power(\condset\times X)$: $\gamma_X(k,U) = \{k\}\times U$; concretely, $\bar FX = \power X$ and $\bar Ff(k, U) =\{f(k,x)\mid x\in U\}$ for $X\rTo^f Y\in\coklcat{G}$. A CTS modelled as a coalgebra $X\rTo^\alpha\bar FX \in \coklcat{G}$ is a function $\condset
  \times X\rTo^\alpha \power X$ that assigns to each state $x\in X$
  and each condition $k\in \condset$, the successors of $x$ under
  condition $k$. 

As mentioned earlier in Section~\ref{sec:indMorKleisli}, it is easier (than in the Kleisli case) to lift a predicate lifting $\contrapower \rTo^\sigma \contrapower F$ of $F$ to define a predicate lifting $\contrapower \forget{\_} \rTo^\lambda \contrapower \forget{\bar F\_}$ of $\bar F$. In particular, $\lambda$ is given by the composition in \eqref{eq:liftOfPlift} and, 
moreover, Theorem~\ref{th:lambda-indexed-morphism} ensures that $\lambda$ is indeed an indexed morphism once we have fixed the predicate lifting $\sigma$ of $F$. 
To this end, we simply take $\sigma$ that corresponds to the box modality (cf.\thinspace Example~\ref{ex:standard_indcat}).
To answer whether these definitions give the right kind of `box' modality for CTSs, let us first instantiate  $\lambda_X$ for any $U \subseteq \condset\times X$:
\begin{equation}\label{eq:boxCTSmodal}
  \lambda_X U
  = \inverse{\gamma}_X \sigma_X U
  = \inverse{\gamma}_X \{U' \mid U' \subseteq U\}
  = \{(k,U') \mid \{k\} \times U' \subseteq U\}.
\end{equation}
Now given a coalgebra $\condset\times X \rTo^{\alpha} \power X\in\set$, we derive the interpretation of box modality for CTSs in the following way (below $x \step k x' \iff x' \in \alpha x k$):
\begin{equation}\label{eq:derivedboxCTSmodality}
  \inverse{\forget{\alpha}} \lambda_X U=\{(k,x) \mid \forall_{x'}\ x \step{k} x' \implies (k,x')\in U\} .
\end{equation}

\paragraph*{Conditional bisimilarity through relation lifting}~\\
\noindent
Next we introduce conditional bisimilarity in which two states might be bisimilar for all conditions or only under certain conditions.
\begin{definition}\label{def:condbisim}
  Given a CTS $X \rTo^\alpha \bar \power X\in \coklcat{G}$, a \emph{conditional bisimulation} is a  relation $R \subseteq (\condset \times X) \times (\condset \times X)$ satisfying:
  \begin{description}
    \item[1]\label{def:filtering} $\forall_{k,k'\in\condset,x,x'\in X}\ (k,x) \mathrel R (k',x') \implies k=k'$.
    \item[2] $\forall_{x_1,x_2,x_3,k} \big(x_1 \step{k} x_3 \land (k,x_1) \mathrel {R} (k,x_2)  \big) \Rightarrow \exists_{x_4\in X} \big(x_2 \step{k} x_4 \land (k,x_3) \mathrel{R} (k,x_4) \big)$.
  \end{description}
  Two states $x,x'\in X$ are \emph{conditional bisimilar under $k$} iff there is a conditional bismulation relation $R$ such that $(k,x) \mathrel {R} (k,x')$.
  Moreover, two states $x,x'$ are \emph{conditional bisimilar}, denoted $x \bisim_X x'$, iff $x$ and $x'$ are conditional bisimilar under every condition $k\in\condset$.
\end{definition}
In order to capture conditional bisimilarity, we first need a fibration $\Psi$ of binary relations on the state space. The first choice for $\Psi$ is to consider the set of all binary relations on the underlying state space, i.e., $\Psi=\contrapower(\forget{\_}\times\forget{\_})$.
Unfortunately, Assumption~\ref{assum:quotient} fails to hold which we explain next.
\begin{remark}\label{rem:CTScounterexample}
  We argue that $\Eq$ cannot have a left adjoint since it does not preserve finite limits (in particular, terminal objects). Clearly, $\iota 1 = 1=\{\bullet\}$ is the terminal object in $\coklcat{G}$ because $\iota$ is the right adjoint of $\forget{\_}$. Now suppose $\Eq 1=(1,=_{\condset\times 1})$ is the terminal object in $\elements{\Psi}$. Then, for any $(X,R)$, there is a unique arrow $(X,R) \rTo^{!_X} \Eq 1$, i.e., $X \rTo^{!_X} 1\in\coklcat{G}$ and $R\subseteq \inverse{(\forget{!_X} \times \forget{!_X})}=_{\condset\times 1}$. But we argue that $!_X$ is \emph{not} a map in $\elements{\Psi}$ for $|\condset|\geq 2$. To see this, let $k,k'\in\condset$ with $k\neq k'$ and let $R\subseteq (\condset\times X) \times (\condset \times X)$ be an equivalence relation such that $(k,x) \mathrel R (k',x)$. Then we find a contradiction
   \[  R \subseteq \inverse{(\forget{!_X} \times \forget{!_X})}=_{\condset\times 1} 
    \implies (k,!_X(k,x)) =(k',!_X(k',x)) \implies k=k'.\]
\end{remark}

So we need to restrict ourselves to relations satisfying the first property of conditional bisimulation (see Definition~\ref{def:condbisim}). An elegant way to address this is by working with the indexed category of abstract relations given in \eqref{eq:indexedcatRelations}; thus, enabling the applicability of constructions given in Section~\ref{sec:indMorKleisli}. Note that binary products $\otimes$ exist in $\coklcat{G}$ and is given by the Cartesian product of two sets. Thus $\Psi$ in \eqref{eq:indexedcatRelations} is well defined. 
\begin{restatable}{lemma}{coKleisliWPbk}\label{prop:coklWPbk}
Assumptions~\ref{assum:quotient} and \ref{assum:projpullback} are valid.
\end{restatable}
Now to invoke the definition of $\bar\lambda$ in \eqref{eq:liftofRlift}, it suffices to define a relation lifting $\sigma$ of the set endofunctor $F=\power$. We take $\sigma$ to be the relation lifting associated with bisimulation relations as defined in Example~\ref{ex:LTS}.

\begin{restatable}{theorem}{coKleisliRelLift}\label{prop:coKlRelLift}
Alternatively, the indexed morphism $\Psi \rTo^{\bar\lambda} \Psi \bar\power$ given in \eqref{eq:liftofRlift} can be defined as follows: $\bar\lambda_X R=\{(k,U, U') \mid \forall_{x\in U}\exists_{x'\in U'} (k,x,x')\in R \land \forall_{x'\in U'}\exists_{x\in U} (k,x,x')\in R\}$.
Moreover, $\bar\lambda$ satisfies Assumption~\ref{assum:presEq}.
\end{restatable}
\begin{corollary}\label{cor:ctsbisim}
  Let $X \rTo^\alpha \power X\in\coklcat{G}$ be a CTS. A relation $R$ on $\condset \times X$ is a conditional bisimulation iff $\inverse{\langle \pi_1^X,\pi_2^X \rangle} R$ is an $\alpha^*\circ\bar\lambda$-coalgebra in $\Psi X$. Moreover, for any $(X,\alpha) \rTo^f (Y,\beta) \in \coklcat{G}$ we have $\bisim_X = \inverse{(\forget{f}\times\forget f)} \bisim_Y$; thus, there is a functor
  $\coalg{\bar{\power}}{\coklcat{G}} \rTo^{\1^{\bar\lambda}} \coalg{\bar{\power}_{\bar\lambda}}{\elements{\Psi}}.$
\end{corollary}
\paragraph*{Modal characterisation of conditional bisimilarity}
\begin{diagram}[width=5em,midshaft]
  \coklcat{\condset\times\_} & \pile{\rTo~{\forget{\_}} \\ {\qquad \scriptscriptstyle\perp} \\ \lTo~{\iota}} & \set & \pile{\rTo~{\sfunctor} \\ {\scriptscriptstyle\perp}\\ \lTo~{\tfunctor} } & \op\ba
\end{diagram}
Consider the above adjoint situations where the adjoint situation on the right is the well known duality (see, for instance \cite{CoalgebraicLogicSurvey_KupkePattinson2011}) between $\set$ and the opposite category of Boolean algebras $\ba$; $\sfunctor$ is the contravariant powerset functor $\contrapower$ and $\tfunctor$ maps a Boolean algebra to its set of ultrafilters. We follow \cite{JacobsSokolova_expressivity_modal_Logic} and use the proposed syntax functor $\ba \rTo^L \ba$ and the interpretation $\power \tfunctor \rTo^\delta \tfunctor L$ induced by the box modality on (unlabelled) transition systems. Since $\bar\power$ is a coKleisli extension of $\power$, i.e., $\bar\power \circ \iota =\iota \circ \power$, we consider the following logical interpretation for CTSs:
$
\bar\power \circ \iota \circ T = \iota \circ \power \circ T \rTo^{\iota\delta} \iota \circ T \circ L.
$

\begin{corollary}
  Since the lifting $\bar\lambda$ preserves equalities, the logic $(L,\iota\delta)$ defined above is adequate for conditional bisimilarity. Moreover, since $\delta$ is injective in each component \cite{JacobsSokolova_expressivity_modal_Logic}, the function $\forget{\iota\delta_{\mathcal A}}$ is injective; thus $(L,\iota\delta)$ is expressive.
\end{corollary}
\section{Conclusions}\label{sec:conclusion}
To recapitulate, we gave a systematic way to construct both predicate
and relation liftings in (co)Kleisli categories and Eilenberg-Moore
categories. 
Relation liftings form the basis to
define behavioural equivalence as a coalgebra of certain lifted
endofunctor in the fibre of relations, although in some cases (such as CTSs) such fibres can be subtle to define.

Once behavioural equivalence is captured as a fibred notion, the Kupke
and Rot setup becomes applicable to obtain its corresponding logical
characterisation. In particular, we gave a recipe to find the left
adjoint $\bar \sfunctor $ of $\bar \tfunctor$ which is a sufficient
condition for both adequacy and expressivity. For coKleisli and
Eilenberg-Moore categories, the construction
\eqref{eq:KupkeRot_behequivSetup} of $\bar \sfunctor$ is based on the
existence of coequalisers in the underlying categories, while in the
Kleisli case one has to resort to a reflective subcategory having
coequalisers (cf.\thinspace
Theorem~\ref{thm:KleisliQu}).

In the future, we plan to develop extend the 2-player game \cite{km:bisim-games-logics-metric} to coalgebras with side effects. 
Lastly, we also plan to investigate whether the given recipe of constructing predicate/relation liftings can be extended to more general monads (like the ones on pseudometric spaces). This should help in developing quantitative modal logics for coalgebras with side effects; thus, providing a pertinent litmus test for the categorical unification of quantitiative expressivity as claimed in the recent work \cite{Katsumata2021_QuantitativeModLogic}.

\bibliographystyle{splncs04}
\bibliography{reference}
\newpage
\appendix
\section{Omitted proofs}
\EqFuncExistence*
\begin{proof}
Let $X \rTo^f Y \in \cat C$. We need to show that $\equiv_X \subseteq f^* \equiv_Y$. First consider the the following diagram, where $p_i^\Box$ and $\pi_i^\Box$ (for $i\in\{1,2\}$ and $\Box\in\{X,Y\}$) are the obvious projection arrows in $\set$ and $\cat C$, respectively.
\begin{diagram}[width=3em,height=2em]
  \forget {X \otimes X} & & \rExistsMap{\langle\forget{\pi_1^X},\forget{\pi_2^X}\rangle} & & & \forget X \times \forget X\\
  &\rdTo(3,2)_{\forget{\pi_i^X}}& && \ldTo_{p_i^X} &\\
  \dExistsMap^{\forget{f\otimes f}} && & \forget X&& \dExistsMap_{(\forget f \times \forget f)}\\
  \forget {Y \otimes Y} & & \rExistsMap^{\langle\forget{\pi_1^Y},\forget{\pi_2^Y}\rangle} && & \forget Y \times \forget Y\\
  &\rdTo(3,2)_{\forget{\pi_i^Y}} & & \dTo_{\forget f}& \ldTo_{p_i^Y} &\\
  && &\forget Y&&
\end{diagram}
Note that all the dashed arrows exists due to the universal property of products. Moreover, it is well known that projection functions $p_i^Y$ are jointly monic. So it suffices to show that the corresponding outermost paths for each $i$ from $\forget{X\otimes X}$ to $\forget{Y}$ commutes, which follows from diagram chasing. Now the result follows from the following implications:
\begin{align*}
  =_{\forget X} \subseteq \inverse{(\forget f \times \forget f)} =_{\forget Y} & \Rightarrow \inverse{\langle\forget{\pi_1^X},\forget{\pi_2^X}\rangle}\! =_{\forget X} \subseteq \inverse{\langle\forget{\pi_1^X},\forget{\pi_2^X}\rangle} \inverse{(\forget f \times \forget f)} =_{\forget Y}\\
  \Rightarrow &\ \inverse{\langle\forget{\pi_1^X},\forget{\pi_2^X} \rangle} =_{\forget X}\subseteq \inverse{\forget{f\otimes f}} \inverse{\langle\forget{\pi_1^Y},\forget{\pi_2^Y} \rangle}=_{\forget Y}.
\end{align*}
The second implication holds due to the commutative of rectangle drawn with dashed edges. \qed
\end{proof}

\BehEquivIsCoalg*
\begin{proof}
  Let $(X,\alpha) \rTo^f (Y,\beta) \in \coalg{F}{\cat C}$ be a coalgebra homomorphism and 
  consider the following commutative diagram.
  \begin{diagram}[width=4em,height=3em]
      \Phi X & \rTo^{\lambda_X} & \Phi FX & \rTo^{\alpha^*}& \Phi X\\
      \uTo^{f^*} && \uTo_{(Ff)^*} && \uTo_{f^*}\\
      \Phi Y & \rTo_{\lambda_Y} & \Phi FY & \rTo_{\beta^*}& \Phi Y
    \end{diagram}
  Then using Assumptions \ref{assum:product}, \ref{assum:presEq} and Proposition~\ref{prop:eqfunexists} we find that $\equiv_Y \subseteq \beta^* (\equiv_{FY} ) = \beta^* \lambda_{Y} (\equiv_Y)$. Furthermore, we derive
  \[
  f^*\equiv_Y \ \subseteq f^* \beta^* \lambda_Y \equiv_Y \implies f^*\equiv_Y \ \subseteq \alpha^* \lambda_X f^*\equiv_Y. \qquad \qed
  \]
\end{proof}

\FibCoalgIndBehEquiv*
\begin{proof}
  Let $R\in\Psi X$ such that $R \preceq \alpha^* \lambda_X R$ and let $\kappa$ be the unit of adjunction $\Qu \dashv \Eq$. For a given coalgebra $X \rTo^\alpha FX$, we are going to first construct a coalgebra of the form $\Qu(X,R) \rTo^{\alpha_R} F\Qu(X,R)$ such that $p\kappa_X$ is the witnessing coalgebra homomorphism. To define $\alpha_R$, consider the following diagram in $\elements{\Psi}$ without the dashed arrow.
        \begin{equation}\label{eq:QuotientCoalg}
        \begin{diagram}[height=3em,width=5em]
          (X,R) & \rTo^{\kappa_X} & \Eq\circ\Qu(X,R)\\
          \dTo^{\alpha} && \dExistsMap_{\Eq (\alpha_R)}\\
           (FX,\lambda_X R) & \rTo_{F\kappa_X} & \Eq \circ F \circ \Qu(X,R)
        \end{diagram}
        \end{equation}
  Note that by the assumption we have $\lambda_X (\equiv_X) =\ \equiv_{FX}$ and, since the unit $\kappa_X\in\elements{\Psi}$, we get that $R \preceq \kappa_X^* \equiv_{\Qu(X,R)}$. Thus,
  \begin{align*}
    \lambda_X R \preceq \lambda_X \kappa_X^* (\equiv_{\Qu(X,R)}) \implies&\ \lambda_X R \preceq (F\kappa_X)^* \lambda_{\Qu(X,R)} \equiv_{\Qu(X,R)}\\
    \implies &\ \lambda_X R \preceq (F\kappa_X)^* \equiv_{F\Qu(X,R)} \enspace.
  \end{align*}
  Thus $F\kappa_X$ drawn in \eqref{eq:QuotientCoalg} is a valid arrow in $\elements{\Psi}$. And the universal property of the unit $\kappa_X$ gives a unique arrow $\Qu(X,R) \rTo^{\alpha_R} F\Qu(X,R)$, i.e., a coalgebraic structure on $\Qu(X,R)$, such that Square \eqref{eq:QuotientCoalg} commutes. Clearly, appyling the forgetful functor $p$ ensures that $p\kappa_X$ is a coalgebra homomorphism. \qed
\end{proof}
\IndexFinal*
\begin{proof}
  Define $\1X=(X,1_X)$ where $1_X$ is the final object in $\Phi X$. Note for any $X \rTo^f Y$, we have $1_X \rTo^{\id{1_X}} 1_X = f^*1_Y $ since the reindexing functor preserves the indexed final objects. Thus, we let $\1f = (f,\id{1_X})$. Next we show that $p \dashv \1$.
  \[
  \infer={X \rTo^f Y \in\cat C}
    {
    \infer={
    p(X,U) \rTo^f Y \in \cat C
    }
    {
    \infer={
    (X,U) \rTo^{(g,\bar g)} \1Y \in \elements{\Phi}
    }
    {
    (X,U) \rTo^{(g,\bar g)} (Y,\id{1_Y}) \in \elements{\Phi}
    }
    }
    }
  \]
  \fbox{$\Uparrow$} Let $X \rTo^f Y \in\cat C$. Then take $g=f$ and let $\bar g$ be the unique arrow $U \rTo^{\bar g} 1_X \in \Phi X$ to the final object $1_X=f^*1_Y$.

\noindent
  \fbox{$\Downarrow$} Take $f=g$ and the implication follows from the construction of $\elements{\Phi}$. \qed
\end{proof}
\FibBehEquivImpLogEquiv*
\begin{proof}
  For adequacy, we show the conditions of \cite[Theorem~18]{KupkeRot_ExprLog_2020} are satisfied. It remains to show that there is a natural transformation $(F\tfunctor A,\equiv_{F\tfunctor A}) = F_{\lambda} \bar \tfunctor A \rTo^{\bar\delta_A} \bar \tfunctor L A = (\tfunctor LA, \equiv_{\tfunctor LA})$. We can simply let $\bar\delta = \delta$ because $\Eq$ is a functor. For expressivity, we show that the conditions of \cite[Theorem~19]{KupkeRot_ExprLog_2020} are satisfied, i.e., the arrow $(FTA,=_{\forget{T\mathcal A}}) \rTo^{\delta_{\mathcal A}} (TL\mathcal A,=_{\forget{TL\mathcal A}}) \in \elements{\Psi}$ is Cartesian. But this follows directly from the injectivity of $\forget{\delta_A}$. To see this let $X \rTo^f Y$ be an arrow in $\cat C$ such that $=_{\forget{X}} \ =\ \inverse{(\forget{f} \times \forget{f})} =_{\forget{Y}}$. Then we derive
  \begin{align*}
    \inverse{\forget{f\otimes f}} (\equiv_Y) =&\ \inverse{\forget{f\otimes f}} \inverse{\langle\forget{\pi_1^Y},\forget{\pi_2^Y}\rangle} (=_{\forget{Y}})\\
    \stackrel{\eqref{eq:squarepbk}}{=} &\ \inverse{\langle\forget{\pi_1^X},\forget{\pi_2^X}\rangle} \inverse{(\forget{f} \times\forget f)} =_{\forget{Y}} \\
    =&\ \inverse{\langle\forget{\pi_1^X},\forget{\pi_2^X}\rangle}  (=_{\forget{X}})\\
    =&\ \equiv_X.
  \end{align*}
\end{proof}
\KleisliQuotients*
\begin{proof}
  Due to \ref{assum:product} product exists in $\klcat{T}$ and the forgetful functor preserves these products. So we can work with concrete equality $=_{\forget{X}}$ instead of $\equiv_X$ (for $X\in\klcat{T}$) by showing that $\bar\sfunctor\dashv\bar\tfunctor$, i.e., the following correspondence holds
  \[
  \infer={
   \infer={Y \rTo^g \bar \sfunctor(X,R) \in \cat A}
   {\bar\sfunctor(X,R) \rTo^{g} Y \in\op{\cat A}}
  }
  {
  \infer={(X,R) \rTo^{f} (\jmath Y,=_{\forget \jmath Y}) \in \elements{\Psi}}
  {
  X \rTo^f \jmath Y \in\klcat T \land R\subseteq \inverse{(\forget f \times \forget f)}=_{\forget{\jmath Y}}
  }
  }
  \]
\fbox{$\Downarrow$} Let $X \rTo^f \jmath Y\in\klcat{T}$ such that $R\subseteq \inverse{(\forget f \times \forget f)}=_{\forget{\jmath Y}}$. Then we find a unique $Y \rTo^g \reflector X \in \cat A$ as the transpose of $f$ under $\reflector\dashv\jmath$. Moreover, using the counit of $\iota\dashv\forget{\_}$ one can show that if $\forget{f} \circ  p_1 =\forget{f} \circ p_2 \implies f \circ p_1' = f \circ p_2'$. Thus, $\reflector p_1' \circ g = \reflector p_2' \circ g$. And by the universal property of equalisers we find a unique $Y \rTo^{u_g} \bar \sfunctor(X,R)$.

\fbox{$\Uparrow$} Let $Y \rTo^{g} \bar \sfunctor (X,R) \in\cat A$. Then take $f$ as the transpose of $e\circ g$ under $\reflector\dashv \jmath$. To show that $f\circ p_1'=f\circ p_2'$ consider the following commutative diagram in $\klcat{T}$, where $\eta$ is the unit of $\reflector\dashv\jmath$.
\begin{diagram}[width=4em,height=3em]
  \jmath\reflector \iota R & \parrows{\jmath\reflector p_1'}{\jmath\reflector p_1'} & \jmath\reflector X & \rTo^{\jmath(e\circ g)} & \jmath Y\\
  \uTo^{\eta_{\iota R}} && \uTo_{\eta_X}\\
  \iota R & \parrows{p_1'}{p_2'}& X
\end{diagram}
\end{proof}
\thetamuComp*
\begin{proof}
  Let $C \rTo^f D \in \klcat{T}$. Then we need to show that
  the following square on the left commutes. But this follows immediately by the commutative diagram drawn on the right, where the top square commutes due to the naturality of $\gamma$.

  \vspace{-0.4cm}
  \begin{minipage}{0.4\textwidth}
    \begin{diagram}
    \forget{\bar FC} & \rTo^{\gamma_C} & G \forget{C}\\
    \dTo^{\forget{\bar Ff}} & & \dTo_{G\forget{f}}\\
    \forget{\bar FD} & \rTo^{\gamma_D} & G \forget{D}
  \end{diagram}
  \end{minipage}
  \begin{minipage}{0.6\textwidth}
    \begin{diagram}[nohug,width=4em,height=1.75em]
    TFC & \rTo^{\gamma_C } & GTC\\
    \dTo^{TFf} & & \dTo_{GTf}\\
    TFTD & \rTo^{\gamma_{TD}} & GTTD\\
    \dTo^{T\theta_D } & & \\
    TTFD & & \dTo_{G\mu_D} \\
    \dTo^{\mu_{FD}} &&\\
    TFD & \rTo^{\gamma_D }& GTD
  \end{diagram}
  \end{minipage}
  \vspace{0.1cm}

  For the converse, take $f=\id{TC}$ and view it as a Kleisli arrow $TC \rTo C$.
\end{proof}
\predliftnda*
\begin{proof}
  Let $V\subseteq Y$ and $X \rTo^f Y$ be a function. Then,
{\allowdisplaybreaks  \begin{align*}
    &\inverse{(f^\act\times 2)} \sigma_Y^a V = \inverse{(f^\act\times 2)}\{(q,b)\in Y^\act\times 2 \mid qa\in V\} \\
    =&\ \left\{(p,b')\in X^\act\times 2 \mid (f^\act\times 2) (p,b') \in \{(q,b)\mid qa\in V\} \right\}\\
    =&\ \left\{(p,b')\in X^\act\times 2 \mid (f\circ p,b') \in \{(q,b) \mid qa\in V\} \right\}\\
    =&\ \{(p,b) \mid f (pa)\in V\}\\
    =&\ \{(p,b) \mid pa \in\inverse f V\} \\
    =&\  \sigma_X^a \inverse{f} V .
  \end{align*}}
  For the mapping associated with termination, we derive
{\allowdisplaybreaks  \begin{align*}
    &\inverse{(f^\act\times 2)} \sigma_Y^\term V = \inverse{(f^\act\times 2)}\{(q,1) \mid q\in Y^\act\} \\
    =&\  \left\{(p,b)\in X^\act\times 2 \mid (f^\act\times 2) (p,b) \in \{(q,1)\mid q\in Y^\act\} \right\}\\
    =&\  \left\{(p,1)\in X^\act\times 2 \mid f\circ p \in Y^\act \right\}\\
    =&\ \sigma_X^\term \inverse{f} V. \qquad \qed
  \end{align*}}
\end{proof}

\GRelLift*
\begin{proof}
  Let $R,R'$ be two binary relations on $X$ such that $R \subseteq R'$. We first show that $\bar\sigma_X R \subseteq \bar\sigma_X R'$. Let $(p,b) \mathrel{\bar\sigma_X R} (p',b)$. Then we find $\forall_{a\in\act}\ (pa \mathrel{R} p'a \implies pa \mathrel{R'} p'a)$. Thus, $(p,b) \mathrel{\bar\sigma_X R} (p',b)$.

  Lastly, we need to show that the equation $\bar\sigma_{X} \circ \inverse{(f\times f)} = \inverse{(Gf \times Gf)} \circ \bar\sigma_Y$ holds. Below $S$ is a relation on $Y$.
  \begin{align*}
    (p,b) \mathrel{\bar\sigma_{X} \circ \inverse{(f\times f)} S} (p',b') &\iff\ b=b' \land \forall_{a}\ pa \mathrel{\inverse{(f\times f)}} p'a\\
    \iff&\ b=b' \land \forall_{a}\ fpa \mathrel R fp'a\\
    \iff&\ (f\circ p,b) \mathrel{\bar\sigma_Y S} (f\circ p',b')\\
    \iff&\ Gf (p,b) \mathrel{\bar\sigma_Y S} Gf(p',b')\\
    \iff&\ (p,b) \mathrel{\inverse{(Gf\times Gf)} \circ \bar\sigma_Y S} (p',b') \enspace. \qed
  \end{align*}
\end{proof}

\LambdaCofilter*
\begin{proof}
  Let $R_i$ (for $i\in\mathbb N$) be relations on $X$. Then
  \begin{align*}
    (p,b) \mathrel{\left(\sigma_X^G \left(\bigcap_{i\in\mathbb N} R_i\right)\right)}  (p',b') \iff&\ b=b' \land \forall_{a\in\act}\ pa \mathrel{(\bigcap_{i\in\mathbb N} R_i)} p'a\\
    \iff&\ b=b' \land \forall_{a\in\act,\mathbb N}\ pa \mathrel{R_i} p'a\\
    \iff&\ \forall_{i\in\mathbb N}\ \left( b=b' \land \forall_{a\in\act}\ pa \mathrel{R_i} p'a \right)\\
    \iff&\ \forall_{i\in\mathbb N}\ (p,b) \mathrel{(\sigma_X^G R_i)}  (p',b')\\
    \iff& (p,b) \left(\mathrel{\bigcap_{i\in\mathbb N}\ \sigma_X^G R_i }\right) (p',b') \enspace.
  \end{align*}
  ote that $\lambda$ is the composition $\inverse{(\gamma_X \times\gamma_X)} \circ \sigma^G$ and since inverse functions preserve arbitrary intersections, we can conclude that $\lambda$ also preserves intersections at each component.

  Lastly, to show that $\bar\lambda$ preserves equalities, let $R = =_{\power X}$. Let $\bar U,\bar U'\subseteq \nda X$ such that $\bar U \mathrel {\lambda_X R} \bar U'$. Clearly, $\bullet \in\bar U \iff \bullet\in\bar U'$. So consider $(a,\bar x)\in\bar U$. Then $\bar x\in\{x\mid (a,x)\in \bar U\}$ and since the two sets $\{x\mid (a,x)\in \bar U\},\{x\mid (a,x)\in \bar U'\}$ are the same, we find that $(a,x)\in\bar U'$. Likewise, we can show the other direction and we conclude that $\bar U =\bar U'$.
\end{proof}
\NDALangEquivFibred*
\begin{proof}
  Let $X \rTo^\alpha \nda X \in \setrel$ be an NDA. Since the relation lifting $\bar\lambda$ preserves intersections, so from Kleene fixed point theorem, the largest fixpoint $\sim_X \in \Psi X$ of $\alpha^* \circ \bar\lambda$ exists. That is, the following sequence stabilises (below $1_{X}\in\Psi X$ is the total relation on $\power X$):
  \[
1_{X}
    \mathrel{\supseteq}
     \inverse{(\forget\alpha\times\forget\alpha)} \circ \lambda_X (\top_X) \mathrel{\supseteq}
      (\inverse{(\forget\alpha\times\forget\alpha)} \circ \lambda_X)^2 (\top_X) \mathrel{\supseteq} \cdots
\]
Let us denote the above projective sequence as a diagram $\cat{N} \rTo^D \Psi X$, i.e. $Di = (\inverse{\forget{\alpha}} \circ \lambda_X)^i \top_X$ (note that the zeroth iteration of a function is assumed to return the full relation $\top_X$ and $\cat N$ has the set $\mathbb N$ of natural numbers as objects and arrows induced by the less-than-equal-to relation). So $\sim_X = \bigcap_i Di$ and note that by definition $\sim_X$ is $\alpha^* \circ \bar\lambda$-coalgebra. 

Let $U_{w}$ is the unique state reached in the determinised system from the state $U$ after performing the trace $w$. Note that for any $i>1$ we have:
\[
U \mathrel{D(i+1)} U' \iff (U \term \iff U' \term) \ \land \ \forall_{a\in\act}\ U_a \mathrel{Di} U_a'.
\]
Moreover, it is not hard to show by induction on $w$ that we have
\[
U \simeq_X U' \quad \iff \quad \forall_{w\in\fseq\act}\ U_{w}\term \iff U'_{w} \term.
\]

Next we show that $\sim_X = \simeq_X$. To show that $\simeq_X\subseteq \sim_X$, it suffices to show that $\simeq_X$ is a $\alpha^*\circ\bar\lambda_X$-coalgebra, i.e., for any $U,U'\subseteq X$, if $U \simeq_X U'$ then $\forget{\alpha} U \mathrel{\lambda_X \simeq_X} \forget{\alpha} U'$. So let $U\simeq_X U'$, clearly, $U\term \iff U'\term$. Moreover, $U_a \simeq_X U_a'$ (for any $a$) because $U\simeq_X U'$. So $\simeq_X\subseteq \sim_X$.

For the other direction, let $U \sim_X U'$. Then for any $w\in\fseq\act$ we have $U \mathrel{D(\#w+1)} U'$, where $\#w$ is the length of the word $w$. Thus, $U_w \term \iff U'_w \term$. Therefore, $U\simeq_X U'$.

Lastly, let $(X,\alpha) \rTo^f (Y,\beta)\in\coalg{\bar F}{\setrel}$, i.e., $X \rTo^f Y$ is a relation satisfying the following two transfer properties:
\begin{itemize}
  \item $\forall_{x,y',a}\ \big(\exists_{y}\ (x \mathrel f y \land y \step a y') \iff \exists_{x'}\ (x' \mathrel f y' \land x \step a x') \big)$.
  \item $\forall_{x}\ \big( x\term \iff \exists_{y}\ (x \mathrel f y \land y\term)\big)$.
\end{itemize}
Note that the second transfer property ensures that $U_w \term \iff (\forget fU)_w \term$, for any $U\subseteq X,w \in\fseq\act$. This is because
\[
U_w \term \iff \exists_{x\in U_w}\ x\term \iff \exists_{x\in U_w,y}\ x\mathrel f y \land y\term \iff (\forget{f}U_w) \term.
\]
Thus, for any $w$ we have $(U_w \term \iff U'_w \term) \iff \big( (\forget{f}U)_w \term \iff (\forget{f} U')_w \term \big)$. Hence $\simeq_X = \inverse{(\forget{f} \times\forget{f})} \simeq_Y$.
\end{proof}

\ThetaNDA*
\begin{proof}
We show that the following square commutes
for any $X \rTo^f Y \in \setrel$.
      \begin{diagram}[width=2cm]
        \act\times\power X + 1 & \rTo^{\theta_X} & \power (\act\times X +1 )\\
        \uTo^{L\reflector f} & & \uTo_{\reflector \bar Ff}\\
        \act\times\power Y + 1 & \rTo^{\theta_Y} & \power (\act\times Y +1 )
      \end{diagram}
  Let $a\in\act,V \subseteq Y$. Then we find
  {\allowdisplaybreaks
      \begin{align*}
        \reflector\bar (Ff) \theta_Y (a,V) =&\ \reflector(\bar Ff) (\{a\} \times V)\\
        =&\ \{(a',x) \mid \exists_{y\in V}\ (a',x) \mathrel{Ff} (a,y) \}\\
        =&\ \{(a,x) \mid \exists_{y\in V}\ x \mathrel f y \}\\
        =&\ \{a\} \times \{x \mid \exists_{y\in V}\ x \mathrel f y \}\\
        =&\ \{a\} \times \reflector f (V)\\
        =&\ \theta_X (a,\reflector f (V)) = \theta_X \circ L(\reflector sf) (a,V).
      \end{align*}}
  Likewise, the case for $\bullet\in 1$ can be verified; thus, $L\sfunctor \rTo^\theta \sfunctor F$ is a natural transformation.
\end{proof}
\coKleisliWPbk*
\begin{proof}[\ref{assum:projpullback}]
  Consider the square drawn in \eqref{eq:squarepbk} and
 let $k,k',y,y',x,x'$ be such that $\langle\forget{\pi_1^Y},\forget{\pi_2^Y}\rangle (k,y,y') = \forget{f} \times \forget{f} ((k,x),(k',x'))$. I.e., $(k,y)=(k,f(k,x))$ and $(k,y')=(k',f(k',x'))$. Thus, we find that $k=k'$. Therefore, we can conclude that $(k,x,x') \in \forget{X \otimes X}$. Moreover, we also find that $\forget{f\otimes f} (k,x,x')=(k,f\otimes f(k,x,x')) = (k,f(k,x),f(k,x')) = (k,y,y')$. Lastly, $\langle\forget{\pi_1^X},\forget{\pi_2^X}\rangle(k,x,x')=((k,x),(k,x')$. I.e., the universal property of weak pullback in $\set$ holds. \qed
\end{proof}
\begin{proof}[\ref{assum:quotient}]
  Let $G=\condset\times\_$ and let $(X,R)\in\elements{\Psi}$, i.e., $X\in\set$ and $R\in\Psi X$. This means we have the following diagram in $\set$:
  \[
  R_! \parrows{p_1}{p_2} GX \rTo^q GX/R_!,
  \]
  where $GX/R_!$ is the coequaliser of $p_1,p_2$ and $R_!=\langle\pi_1^X,\pi_2^X\rangle_! R=\{((k,x),(k,x')) \mid (k,x,x')\in R\}$. So we define $\Qu(X,R)$ to be the object $\iota (GX/R_!)$. In other words, $\Qu$ can be defined by the composition:
  \[
  \elements{\Psi} \rTo^{\forget{\_}_!} \elements{\Psi'} \rTo^{\Qu'} \set \rTo^{\iota} \coklcat{G},
  \]
  where $\Psi'$ is the standard indexed category of relations on sets and $\Qu'$ is the usual quotient functor as described in Remark~\ref{rem:quotientSets}. The functor $\forget{\_}_!$ is defined as: $(X,R) \mapsto (\forget{X},R_!)$ and $f \mapsto \forget{f}$. This is welldefined due to Beck-Chevalley property on the indexed category $\contrapower$, which holds for weak pullback squares. In particular, let $(X,R) \rTo^f (Y,S) \in \elements{\Psi}$, i.e., $X \rTo^f Y \in \cat C$  and $R \subseteq \inverse{(f\otimes f)}S$. By applying the direct image functor we get
  \[
  \langle\forget{\pi_1^X},\forget{\pi_2^X}\rangle_! R \subseteq \langle\forget{\pi_1^X},\forget{\pi_2^X}\rangle_! \inverse{(f\otimes f)}S = \inverse{(\forget{f}\times\forget{f})} \langle\forget{\pi_1^X},\forget{\pi_2^X}\rangle_! S .
  \]
  Note that the righmost equality holds due to Beck-Chevalley condition and \ref{assum:projpullback}. Therefore $\Qu$ is a functor.

  Next we show that $\Qu \dashv \Eq$, i.e., we establish the following correspondence:
  \[
  \infer={
   \infer={X \rTo^g Y \in \coklcat{G} \land R \subseteq \inverse{(\forget{g \otimes g})}\equiv_Y}
   {(X,R) \rTo^g \Eq Y=(Y,\equiv_Y) \in \elements{\Psi}}
   }
  {\iota (GX/R_!) \rTo^f Y \in \coklcat{G}}
  \]
  \fbox{$\Uparrow$} Let $X \rTo^g Y \in \coklcat{G}$ (i.e., $GX \rTo^g Y\in\set$) and $R \subseteq \inverse{(\forget{g \otimes g})}\equiv_Y$. We claim that $g\circ p_1 =g \circ p_2$. So let $(k,x) \mathrel {R_!} (k',x')$. Then we find $k=k'$ because of the construction of $R_!$ and $g(k,x)=g(k,x')$ because of the inequality $R \subseteq \inverse{(\forget{g\otimes g})}\equiv_Y$. This proves the claim. And from the universal property of coequalisers in $\set$, we find a unique function $GX/R \rTo^{g'} Y$ satisfying $g'\circ q=g$. So fix $f=\iota g'$ and clearly $f$ is uniquely determined by $g$.

  \fbox{$\Downarrow$} Let $\iota (GX/R) \rTo^f Y \in \coklcat{G}$. Recall the unit $X \rTo^{\eta_X} \iota \forget{X} \in \coklcat{G}$ of the adjunction $\forget{\_}\dashv \iota$. Then, let $g$ be the following composition:
  \[
  X \rTo^{\eta_X} \iota \forget{X} \rTo^{\iota q} \Qu(X,R) \rTo^f Y \ \in \coklcat{G}.
  \]
  Clearly, $g$ is uniquely determined by $f$. Now it remains to show that for $k,x,x'$, if $(k,x,x') \in R$ then $\forget{g\otimes g}(k,x,x') = (k,g\otimes g(k,x,x'))=(k,g(k,x),g(k,x'))\in \equiv_Y$. Thus it remains to show that $g(k,x)=g(k,x')$.

  To this end, we first evaluate $g$. Note that $\eta$ is given by the identity arrow and using the definition of coKleisli composition we find that
  \[
  (k,x) \rMapsto^{\Delta_X} (k,k,x) \rMapsto^{G\eta_X} (k,k,x)  \rMapsto^{\iota q} q(k,x).
  \]
  Thus, $\iota q \circ \eta_X = q$. So computing the composition $X \rTo^q \Qu(X,R) \rTo^f Y \in\coklcat{G}$ results in the evaluation of $g$. In particular,
  \[
  g(k,x) = f (Gq (\Delta_X (k,x))) = f (k, q(k,x)).
  \]
  So if $(k,x,x')\in R$ then $q(k,x)=q(k,x')$. Thus, $g(k,x) = f(k, q(k,x)) =f(k,q(k,x')) = g(k,x')$.\qed
\end{proof}

\coKleisliRelLift*
\begin{proof}
  Recall from \eqref{eq:liftofRlift}, the map $\bar\lambda$ is given by the composition:
  \begin{multline*}
    \Psi X \rTo^{\langle\forget{\pi_1^X},\forget{\pi_2^X}\rangle_!} \contrapower(\forget {X} \times \forget X) \rTo^{\sigma_{\forget X}} \contrapower (\power\forget{X} \times \power\forget{X}) \\ \rTo^{\inverse{\gamma_X\times\gamma_X}} \contrapower (\forget{\power X} \times\forget{\power X}) \rTo^{\inverse{\langle\forget{\pi_1^{FX}},\forget{\pi_2^FX}\rangle}} \Psi \bar \power X.
  \end{multline*}
  Let $R\subseteq G(X \times X)$, where $G=\condset\times\_$. Then we compute
  {\allowdisplaybreaks
  \begin{align*}
    \langle\forget{\pi_1^X},\forget{\pi_2^X}\rangle_! R &= \ \{((k,x),(k,x'))\mid (k,x,x')\in R \}\\
    \sigma_{\forget X} \langle\forget{\pi_1^X},\forget{\pi_2^X}\rangle_! R &= \ \Big\{(S,S') \in \power\forget X \times \power \forget X \mid \\
    & \quad \forall_{k,x} \big( k \mathrel S x \implies \exists_{x'}\ k\mathrel {S'} x' \land (k,x,x')\in R \big) \land \\
    & \quad \forall_{k,x'} \big( k\mathrel{ S'} x' \implies \exists_{x}\ k\mathrel S x \land (k,x,x')\in R \big)\Big\}\\
    \inverse{(\gamma_X \times \gamma_X)} (\sigma_{\forget X} \langle\forget{\pi_1^X},\forget{\pi_2^X}\rangle_! R) &= \{((k,U),(k,U')) \mid U,U'\subseteq X \ \land  \\
    \forall_{x\in U}\exists_{x'\in U'}& (k,x,x')\in R \land \forall_{x'\in U'}\exists_{x\in U} (k,x,x')\in R \}\\
    \bar\lambda R= \ \{(k,U, U') \mid \forall_{x\in U}\exists_{x'\in U'}& (k,x,x')\in R \land \forall_{x'\in U'}\exists_{x\in U} (k,x,x')\in R\}.
  \end{align*}
  }
  For \ref{assum:presEq} we first find that
  \[(k,x,x')\in \equiv_X \iff (k,x,x') \in \inverse{\langle \pi_1^X,\pi_2^X \rangle} =_{\forget X} \iff (k,x) = (k,x').\]
  Thus, $\bar\lambda_X$ preserve abstract equalities. \qed
\end{proof}

\section{Linear Weighted Automaton (LWA)}\label{sec:KleisliEx:LWA}
We consider (linear) weighted automata (LWA) as coalgebras as studied in \cite{Jacobs_etal:trace_EM_KL}. LWA are modelled
as coalgebras of the endofunctor $\mats(1+\act\times\_)$, where
$\mathbb F$ is a field and $\mats$ is the multiset monad. The set of functions $\mathbb F^X$ having finite support is $\mats X$ and on arrows it is given by $\mats f(\tau)(y)=\sum_{x\in f^{-1}(y)}\tau(x)$ (for $X \rTo^f Y\in\set$). We write $x\term_s$ and $x \step{a,s'}x'$ whenever $\alpha(x)(\bullet)=s$ and $\alpha(x) (a,x')=s'$ for a given LWA $X \rTo^\alpha \mats (\act\times X + 1)$.

\noindent
Recall the language of a given LWA $\alpha$ starting from a state $x\in X$ is an inductively defined function $\fseq\act\rTo^{\trace (x)}\mathbb F$ described below, where $a\in\act$, $w\in\fseq\act$, and $\eseq$ is the empty word.
\[
\trace (x) (\eseq) = \alpha(x)(\bullet), \qquad\quad \trace (x)(a w) = \sum \{s\cdot \trace(x')(w)\mid x\step {a,s} x'\} .
\]
Two states $x,x'\in X$ are \emph{(weighted) language equivalent} iff $\trace (x) = \trace (x')$. This coincides with
coalgebraic behavioural equivalence  in $\klcat\mats$ (see
\cite{Jacobs_etal:trace_EM_KL,kk:coalgebra-weighted-automata-journal}). Note that probabilistic automata can be encoded by letting $\mathbb F = \mathbb R$ and restricting the weights to the interval $[0,1]$.
\paragraph*{Predicate lifting for weighted automata} We employ a similar simplification as carried out in NDA, by working with the indexed category $\contrapower\circ (\forget{\_} \times \forget{\_})$. Now to apply techniques of Section~\ref{sec:indMorKleisli}, we fix $F=\act\times\_ + 1, G=\_^\act \times \mathbb F$, and recall from \cite[Section~7.3]{Jacobs_etal:trace_EM_KL}
the distributive laws $F\mats X \rTo^{\vartheta_X}
\mats FX$ and $\mats FX \rTo^{\gamma_X}
G\mats X$:
{\allowdisplaybreaks
\begin{align*}
  \quad \vartheta_X(\bullet)(\heartsuit)&=
  \begin{cases}
    1, & \text{if } \ \heartsuit=\bullet\\
    0, & \text{otherwise}.
  \end{cases}
  \quad \vartheta_X(a,\tau)(\heartsuit)=
  \begin{cases}
    \tau(x), & \text{if}\ \heartsuit=(a,x), \text{ for some $x\in X$}\\
    0, & \text{otherwise}.
  \end{cases}
  \\
  \gamma_X(\bar p)&=(\gamma_X^\act \bar p,\bar p(\bullet)), \ \text{where}\ \gamma_X^\act \bar p(a)(x) = \bar p(a,x) \ \text{(for $x\in X,\bar p\in\mats FX$)}.
\end{align*}
}
We know (from \cite{Jacobs_etal:trace_EM_KL}) that $\gamma$ is compatible with $\theta$ and $\mu$ (the multiplication of the monad $\mats$) in the sense of Lemma~\ref{lemma:theta&mu_compatible}.
Similar to NDAs, consider the following predicate liftings $\contrapower X \rTo^{\sigma_X^a,\sigma_X^s} \contrapower(X^\act\times \mathbb F)$ (for $a\in\act$ and $s\in \mathbb F$):
\[U\mapsto \{(p,s)\in X^\act\times \mathbb F \mid p(a) \in U\} \quad\text{and} \quad U \mapsto \{(p,s) \mid p\in X^\act\},\  \text{respectively.}\]
The proof of the following lemma is similar to the proof of Lemma~\ref{lem:predlift_nda}. And thanks to Theorem~\ref{th:lambda-indexed-morphism}, we know that $\inverse{\gamma}\circ \sigma^a$ and $\inverse{\gamma}\circ \sigma^s$ are valid predicate liftings.
\begin{lemma}\label{lem:predlift_gpa}
  The above mappings $\sigma_X^a,\sigma_X^s$ (for $a\in\act,s\in\mathbb F$) are indexed morphisms.
\end{lemma}

\begin{restatable}{lemma}{LWApredlift}\label{lem:LWA_predlift}
  For any $\mathbb U \subseteq\mats X$ we find that $\lambda_X^a(\mathbb U)=\{\bar p \in \mats (\act\times X +1) \mid \gamma_X^\act \bar p(a) \in \mathbb U\}$ and $\lambda_X^s(\mathbb U) =\{\bar p \in \mats( \act\times X +1 ) \mid \bar p(\bullet) =s \}$.
\end{restatable}
\begin{proof}
For any $\mathbb U \subseteq\mats X$ we find that
\begin{align*}
  \lambda_X^a(\mathbb U) =&\ \inverse{\gamma_X}\sigma_{\mats X}^a(\mathbb U)\\
  =&\ \inverse{\gamma_X}\{(p,s) \in (\mats X)^\act \times \mathbb F \mid p(a) \in \mathbb U\}\\
  =&\ \left\{\bar p \in \mats (\act\times X +1) \mid \gamma_X \bar p \in \{(p,s) \in (\mats X)^\act \times \mathbb F \mid p(a)\in \mathbb U\} \right\}\\
  =&\ \left\{\bar p \in \mats (\act\times X +1) \mid \gamma_X^\act \bar p(a) \in \mathbb U \right\}.
\end{align*}
Similarly, in the context of termination, we find (for each $s\in\mathbb F$):
{\allowdisplaybreaks
\begin{align*}
  \lambda_X^s(\mathbb U) =&\ \inverse{\gamma_X}\sigma_{\mats X}^s(\mathbb U)\\
  =&\ \inverse{\gamma_X}\{(p,s) \mid p\in (\mats X)^\act \}\\
  =&\ \left\{\bar p \in \mats( \act\times X +1 ) \mid \gamma_X \bar p \in \{(p,s) \mid p\in (\mats X)^\act\} \right\}\\
  =&\ \left\{\bar p \in \mats( \act\times X +1 ) \mid \bar p(\bullet) =s \right\}.
\end{align*}
}
The proof that $\lambda^a_X,\lambda^s$ preserves finite meet is similar to the Boolean case (cf.\thinspace Lemma~\ref{lem:predlift_nda}).
\end{proof}

\noindent
Just like in our running example, the determinisation of an LWA $\alpha$ is the composition:
\[
\mats X \rTo^{\mats \alpha} \mats\mats (\act\times X + 1) \rTo^{\mu_{\act\times X + 1}} \mats(\act\times X +1) \rTo^{\gamma_X} (\mats X)^\act \times \mathbb F.
\]
More concretely, it maps a $p\in \mats X$ to a pair $(\hat p,s)$,
where $\hat p(a)(x') = \sum_{x\in X} p(x)\cdot \alpha(x) (a,x')$ and
$s = \sum_{x\in X} p(x)\cdot \alpha(x)(\bullet)$. In terms of SOS
rules, determinisation is given as follows:
\[
\infer{p \step a \hat p (a)}{p\in\mats X} \qquad
\infer{p \term_s}{p\in\mats X\quad s=\sum_{x\in X} p(x)\cdot \alpha(x)(\bullet)}
\]
\begin{restatable}{lemma}{LWAsepmod}\label{lem:LWA_sep_modalities}
  For any $\mathbb U\subseteq \mats X$ and $X\rTo^\alpha \mats(\act\times X+1)$, we have $\inverse{\forget{\alpha}}\lambda_X^a \mathbb U=\{p\in\mats X \mid p \step a \hat p (a) \implies \hat p(a) \in \mathbb U\}$ and $\inverse{\forget{\alpha}}\lambda_X^s \mathbb U=\{p\in\mats X\mid s=\sum_{x\in X}\ p(x)\cdot\alpha(x)(\bullet)\}$. 
\end{restatable}
\begin{proof}
  Let $\mathbb U\subseteq \mats X$ and a given LWA $X\rTo^\alpha \mats(\act\times X+1)$. Then we derive
  \begin{align*}
  \inverse{\forget{\alpha}}\lambda_X^a \mathbb U =\ & \inverse{\forget\alpha} \left\{\bar p \in \mats (\act\times X +1) \mid \gamma_X^\act \bar p(a) \in \mathbb U \right\}\\
  =\ & \left\{p\in\mats X \mid \forget{\alpha } p \in \left\{\bar p \in \mathcal D (\act\times X +1) \mid \gamma_X^\act \bar p(a) \in \mathbb U \right\}\right\}\\
  =\ & \left\{ p\in\mats X \mid \gamma_X^\act (\mu_{\act\times X+1} \mats\alpha (p)) a \in \mathbb U\right\}\\
  =\ & \left\{p\in\mats X \mid p \step a \hat p (a) \implies \hat p(a) \in \mathbb U\right\}.
\end{align*}
Similarly, we have a modality to handle termination that can be derived as follows:
\begin{align*}
  \inverse{\forget{\alpha}}\lambda_X^s \mathbb U =\ & \inverse{\forget\alpha} \left\{\bar p \in \mats( \act\times X +1 ) \mid \bar p(\bullet) =s \right\}\\
  =\ & \left\{p\in\mats X \mid \forget \alpha p \in \left\{\bar p \in \mats( \act\times X +1 ) \mid \bar p(\bullet) =s \right\} \right\}\\
  =\ & \{p\in\mats X\mid s=\sum_{x\in X}\ p(x)\cdot\alpha(x)(\bullet)\}.
\end{align*}
\end{proof}

\paragraph*{Weighted language equivalence through relation lifting}
Consider the relation lifting $\bar\sigma$ of $G$ that maps a relation $R$ on $X$ to $\bar\sigma_X R$ defined as: $(p,s) \mathrel{\bar \sigma_X R} (p',s') \iff s=s' \land \forall_{a\in\act}\ pa \mathrel{R} p'a.
$
So \eqref{eq:liftofRlift} gives a relation lifting $\bar\lambda$ for $\bar F$. Concretely, $\bar\lambda$ maps a relation $R$ on $\mats X$ to a relation $\bar\lambda_X R$ on $\mats FX$ given as:
$
\bar p \mathrel{\bar \lambda_X R} \bar p' \ \iff \ \bar p (\bullet) =\bar p'(\bullet) \ \land \ \forall_{a\in\act}\ \gamma_X^\act \bar p (a) \mathrel R \gamma_X^\act \bar p' (a).
$
\begin{restatable}{theorem}{LWALangEquivFibred}\label{thm:LWALangEquivFibred}
  Let $X \rTo^\alpha FX\in\klcat{\mats}$ be an LWA. Then weighted language equivalence $\equiv_X\subseteq \mats X \times \mats X$ on the determinised system is a $\Psi(\alpha)\circ\bar\lambda_X$-coalgebra. 
  Moreover,  $\equiv_X = \inverse{(\forget{f} \times \forget f)} \equiv_Y$ for any coalgebra homomorphism $(X,\alpha) \rTo^f (Y,\beta)$; thus, there is a behavioural conformance functor $\coalg{\bar F}{\klcat{\mats}} \rTo^{\1^{\bar\lambda}} \coalg{\bar F_\lambda}{\elements{\Psi}}$.
\end{restatable}
\begin{proof}
  Let $X \rTo^\alpha FX\in\klcat{\mats}$ be an LWA. Just like in the proof of Theorem~\ref{thm:NDALangEquivFibred}, we will show that weighted language equivalence $\equiv_X$ coincides with the relation $\sim$ defined as the limit of the diagram $\cat N \rTo^D \Psi X$ defined in the proof of Theorem~\ref{thm:NDALangEquivFibred}. Recall the determinisation of an LWA in terms of SOS rules:
  \[
\infer{p \step a p_a}{p\in\mats X} \qquad
\infer{p \term_s}{p\in\mats X\quad s=\sum_{x\in X} p(x)\cdot \alpha(x)(\bullet)},
\]
where $p_a(x') = \sum_{x\in X} p(x)\cdot \alpha(x) (a,x')$. We can extend this notation to words as follows:
$p_\eseq= p$ and $p_{aw}=(p_a)_w$. Thus, $p\equiv_X p' \iff \forall_{w,s}\ p_w \term_s \iff p'_w \term_s$.

  Note that for any $i>1$ and $p,p'\in\mats X$ we have
  \[
  p \mathrel{D(i+1)} p' \iff \forall_{s\in\mathbb F}\  (p \term_s \iff p' \term_s) \ \ \land\ \ p_a \mathrel{Di} p_a'.
  \]
  Clearly, $\equiv_X \subseteq \sim$ because $\equiv_X \subseteq \alpha^*\lambda_X \equiv_X$. For the other direction, suppose $p \sim p'$. Then for any word $w\in\fseq\act$ we have $p \mathrel{D(\#w + 1)} p'$ which results $p_w \term_s \iff p_w' \term_s$, for any $s\in\mathbb F$. Since $w$ was chosen arbitrarily, we find that $p \equiv_X p'$.
\end{proof}

\paragraph*{Logical characterisation of weighted language equivalence}
Just like language equivalence, its weighted variant is a `linear' notion of behavioural equivalence; thus, we set $\cat A=\set$ and $L=F$. As for the dual adjunction, first recall that $\klcat{\mats}$ is equivalent to the Eilenberg-Moore category $\emcat{\mats}=\veccat_{\mathbb{F}}$ \cite{Jacobs_etal:trace_EM_KL}, which is due to the fact that every vector space has a basis. In particular, the comparison functor $\klcat{\mats} \rTo^{\mathcal K} \emcat{\mats}$ is fully faithful and essentially surjective. As a result, $\mathcal K$ has an inverse $\mathcal{K}'$ that maps every vector space to its basis which can be arranged as $\mathcal K \dashv \mathcal K'$. Moreover, the category of vector spaces is related with the category of sets by a dual adjunction \cite{Jacobs_etal:trace_EM_KL} as indicated below.
  \begin{diagram}[width=6em,midway]
    \klcat{\mats} & \pile{\rTo^{\mathcal K} \\ {\scriptscriptstyle \perp} \\ \lTo_{\mathcal K'}} & \emcat{\mats} & \pile{\rTo^{\hom(\_,\mathbb F)} \\ {\scriptscriptstyle \perp} \\ \lTo_{\mathbb F^{\_}}} & \op\set.
  \end{diagram}
So $\sfunctor=\hom(\mathcal K\_,\mathbb F)$ and $\tfunctor=\mathcal K' \mathbb F^{\_}$. Note that for any set $X$ we have $\sfunctor X \cong \mathbb F^X$ and $\mathcal{K}' (\mathbb F^X) \cong X$. The former is true because $\hom(\mats X,\mathbb F) \cong \set(X,\mathbb{F})$, while the latter is because the basis of a vector space generated by a set is isomorphic to the set itself. So $\op\set$ is a subcategory of $\klcat{\mats}$ and Theorem~\ref{thm:KleisliQu} ensures that $\Eq\circ\tfunctor$ has a left adjoint $\bar\sfunctor$. Finally, $\delta_X$ (for any set $X$) should be a natural transformation of type $FTX = FX\rTo^{\delta_X} FX=TLX$. Therefore, we simply take $\delta_X$ to be an identity arrow on $FX$ in $\klcat{\mats}$. Lastly, the theory map $X \rTo^{\theory{\_}_X} \tfunctor\fseq\act\in\klcat{\mats}$ is given by the trace function $\trace$.
\begin{restatable}{theorem}{LogicalWLang}\label{thm:logicalWLang}
  The above defined map $\theory{\_}$ is a theory map for a given LWA. As a result, the logic $(L,\delta)$ is both adequate and expressive for weighted language equivalence.
\end{restatable}
\begin{proof}
  We show the square associated with theory map is commutative, i.e., $\tfunctor h \circ \theory{\_}_X = \delta_{\fseq{\act}} F\theory{\_}_X \circ \alpha$; the uniqueness follows from structural recursion. Since $\delta_{\fseq\act}$ is identity arrow, so it suffices to show that $\theory{\_}_X$ is a coalgebra homomorphism; $\tfunctor h \circ \theory{\_}_X = \bar F\theory{\_}_X \circ \alpha$. Note that we follow \cite{kk:coalgebra-weighted-automata-journal} by using matrix multiplication to denote the composition of arrows in $\klcat{\mats}$. Also, for a function $Y \rTo^f X$, $\mathcal{K}' (\mathbb{F}^f)$ corresponds to unit matrix; namely, the cell $(\delta_x,\delta_y)$ in $\mathbb {F}^f$ corresponds to $1$ if $F^f(\delta_x) = \delta_x \circ f = \delta_y$; $0$ otherwise. So concretely, $\tfunctor f$ is a function defined as follows: $\tfunctor f x y = 1$ if $fy=x$; 0, otherwise.
  {\allowdisplaybreaks
  \begin{align*}
    F\theory{\_}_X\circ \alpha (x,(a,w)) =&\ \sum_{ (a',x')\in\text{supp}(\alpha x)} \alpha(x,(a',x')) \cdot \bar F\theory{\_}_X((a',x'),(a,w))\\
    =&\ \sum_{(a,x')\in\text{supp}(\alpha x)} \alpha(x,(a,x')) \cdot \theory{x'}(w)\\
    =&\ \theory{x}_X(aw)\\
    =&\ \sum_{w'=aw} \theory{x}_X w'\cdot \tfunctor h(w',(a,w))\\
    =&\ \tfunctor h \circ \alpha (x,(a,w)).
  \end{align*}
  }
  Adequacy follows from Theorem~\ref{thm:fibredbehequiv->logequiv}. Furthermore, expressivity follows since $\forget{\delta_{\fseq\act}}$ is injective.
\end{proof}
\section{Eilenberg-Moore categories}\label{sec:EMcats}
Although, in this paper, there are no case studies (see \cite{technicalreport} for predicate liftings in generalised Moore automata) worked out in the setting of Eilenberg-Moore categories, 
it is still worthwhile to report that Eilenberg-Moore categories (at least when $T$ is finitary)
are more well-behaved than Kleisli categories in satisfying the assumptions of this paper. First it is straightforward to show that every predicate lifting $F$ induces a predicate lifting of $\bar F$. 
\begin{proposition}\label{prop:EMCAT_indexedcat}
  Let $\bar F$ be a lifting of $F$ to $\emcat{T}$ with $\cat C \rTo^T \cat C$. Then,
  \begin{enumerate}
    \item every indexed category on $\cat C$ induces an indexed category $\emcat{T}$ by composing with $\forget{\_}$.
    \item every indexed morphism $\Phi \rTo \Phi F$ induces an indexed morphism of type $\Phi \forget{\_} \rTo \Phi \forget{\_} \bar F$.
  \end{enumerate}
\end{proposition}
  Second, 
  the adjunction $\Qu\dashv\Eq$ does exist when $T$ is finitary and $\Phi,\Psi$ are instantiated as in the Kleisli case, i.e., $\Phi=\contrapower(\forget{\_})$ and $\Psi=\contrapower(\forget{\_}\times\forget{\_})$. Intuitively, $\Qu$ creates quotients w.r.t. the smallest congruence relation generated by a relation on the underlying algebra. Formally, for an algebra $(X,h)\in\emcat T$, let $\congset{X,h}$ denote the poset of congruences on $X$ ordered by $\subseteq$. 
  Next we claim that arbitrary meets exist in $\congset{X,h}$ since $\emcat{T}$ is complete. Therefore, we can construct 
  \[
  \mathcal C(R)= \bigwedge \{ (R',h') \mid (R',h')\in\congset{X,h} \land R \subseteq R' \}\enspace.
  \]
  Now the algebraic structure on $\mathcal C(R)$, i.e., a function
  $T\mathcal C(R) \rTo \mathcal C(R)$ exists due to the universal
  property of limits, which we denote simply $h_R$. So we take
  $\Qu((X,h),R)$ to be the coequaliser of
  $
  (\mathcal C(R),h_R) \rTo^{p_i} (X,h)
  $
  in $\emcat{T}$.
\begin{theorem}
  Suppose $T$ has a finite presentation, then the mapping $\Qu$ is left adjoint to $\Eq$ functor.
\end{theorem}
\begin{proof}
  We need to show that the following correspondence holds.
  \[
  \infer={
   \infer={(X,h) \rTo^g (Y,k)\in\emcat T \land R\subseteq \inverse{(g\times g)}=_Y}
   {
   (X,h,R) \rTo^{g} \Eq (Y,k) \in\elements{\Psi}
   }
   }
   {\Qu(X,h,R) \rTo^{f} (Y,k) \in \emcat{T}}
  \]
  The direction from top to bottom is easy. Let $q_R$ be the coequaliser map of type $(X,h) \rTo^{q_R} \Qu(X,h,R)$. Then, we take $g$ to be the arrow $f\circ q_R$.

  For the converse, let $(X,h) \rTo^g (Y,k)\in\emcat T$ be an algebra map such that $R\subseteq \inverse{(g\times g)}=_Y$. Now it suffices to show that $g\circ p_1 = g\circ p_2$ because the universal property of coequaliser then gives the unique arrow $f$.

  To prove this we recall the construction of the congruence $\mathcal C(R)$ generated by $R$ when $T$ has a finite presentation $\mathbb T=(\Sigma,E)$, where $\Sigma$ is a set of function symbols with finite arity and $E$ is a set of equations.
  \begin{itemize}
    \item Define $R_0$ the reflexive and symmetric closure of $R$.
    \item Define $R_{i+1} = R_i \circ R_i \cup \{(f(x_1,\cdots,x_n),f(x_1',\cdots,x_n')) \mid \forall_{j}\  x_j R_i x_j' \}$, where $f$ is an $n$-ary operator in $\Sigma$.
    \item $\mathcal C (R) = \bigcup_{i} R_i$.
  \end{itemize}
  Now by induction we show that $\forall_i\ t \mathrel {R_i} t' \implies \mathbb T\vdash gt = gt'$. The base case is trivial. So assume the property is true for some $i\geq 0$. Let $t \mathrel {R_{i+1}} t'$. Then we identify two cases: either $t \mathrel {R_{i} \circ R_i} t'$ or there is an $n$-ary operator and variables $x_j,x_j'$ (for $1\leq j \leq n$) such that $t= f(x_1,\cdots,x_n)$, $t'=f(x_1',\cdots,x_n')$ and $x_j \mathrel {R_i} x_j'$ (for all $j$). If the former is true then the property holds due to the inductive hypothesis and the transitivity of $\vdash$. So suppose the latter is true.  Then by the inductive hypothesis we find that $\mathbb T\vdash gx_j = gx_j'$ (for each $1\leq j \leq n$). Since $g$ is a homomorphism, so applying the context rule of $\vdash$ gives us
  \[
  \mathbb T \vdash gt=f(gx_1,\cdots,gx_n) = f(gx_1',\cdots,gx_n')=gt'. \qquad \qed
  \]
\end{proof}
\subsection*{Predicate liftings for generalised Moore machines}
In \cite{bonchi_etal_2016_dectraces,Silva_etal:General_Det}, the authors captured failure/ready/trace equivalences as behavioural equivalence for coalgebras living in the Eilenberg-Moore category introduced by a finitary power set monad $\fpower$. Interestingly, these linear equivalences can be seen as an instance of behavioural equivalence induced by a common `generalised Moore' endofunctor $F_S=S \times \_^\act$ \cite[Section~3]{bonchi_etal_2016_dectraces}, where $S\in\emcat{\fpower}$ is an arbitrary semi-lattice.

Given a map $X \rTo^o S\in\emcat{\power_\omega}$, then the generalised determinisation \cite{Bonchi_et_al:UpToFibrations} of an LTS $X \rTo^\alpha (\power_\omega X)^\act$ is a coalgebra $\fpower X \rTo^{(o',\alpha')} F_S (\fpower X)\in\emcat{\fpower}$, where the two functions are given as follows (cf.\thinspace\cite{bonchi_etal_2016_dectraces}):
\[
o' (U) = \bigvee_{x\in U} o(x) \quad\text{and}\quad \alpha' (U)(a) = \bigcup_{x\in U}\ \alpha(x)(a) \quad \text{(for $U\subseteq X,a\in\act$)}\enspace.
\]
In other words, the generalised determinisation induces a transition relation on $\fpower X$ and a predicate on $\fpower X \times S$ defined by the following SOS rules:
\[
\infer {U \step a U_\alpha} {U\subseteq X \quad U_\alpha=\{x'\mid \exists_{x\in U}\ x \step a x'\}} \qquad
\infer {U \term o'(U) }{U\subseteq X \qquad o' (U) = \bigvee_{x\in U} o(x)} \enspace.
\]
The next theorem is recalled from \cite{Silva_etal:General_Det} and captures the aforementioned linear equivalences as instances of behavioural equivalence.
\begin{theorem}
  Two states $U,U'\subseteq X$ in the coalgebra $\fpower X \rTo^{(o',\alpha')} F_S (\fpower X)\in\emcat{\fpower}$ are behaviourally equivalent iff they get mapped to a common point in the final coalgebra (which exists for the Moore functor $F_S$). Moreover,
  \begin{enumerate}
    \item For $S=2$ and $o$ as constant 1, behavioural equivalence coincides with trace equivalence.
    \item For $S=\fpower\fpower \act$ and $o(x)$ (for each $x\in X$) is the set of sets of actions that are refused by the state $x$, behavioural coincides with failure equivalence.
    \item For $S=\fpower\fpower \act$ and $o(x)$ (for each $x\in X$) is the singleton containing the set of actions enabled at $x$, behavioural coincides with ready equivalence.
  \end{enumerate}
\end{theorem}

Next we define two families of indexed morphisms $\contrapower \rTo^{\lambda^a} \contrapower \circ F_S$ (for each $a\in\act$) and $\contrapower \rTo^{\lambda^s} \contrapower \circ F_S$ (for each $s\in S$).
\begin{itemize}
  \item $\lambda^a_X U=\{(s',p)\in S\times X^\act \mid p(a)\in U\}$
  \item $\lambda^s_X U=\{(s,p) \mid p\in X^\act \}$, for $U\subseteq X$.
\end{itemize}
\begin{restatable}{proposition}{MooreIndexedMap}\label{prop:F_S_indexedmorphism}
  The above mappings $\contrapower \rTo^{\lambda^a,\lambda^s} \contrapower \circ F_S$ are indexed morphisms. 
\end{restatable}
\begin{proof}
  Let $V\subseteq Y$ and $X \rTo^f Y\in\set$. Then we derive
  \begin{align*}
    \lambda_X^a\inverse{f}V =&\ \{(s,p) \mid p(a) \in \inverse f V\}\\
    =&\ \{(s,p) \mid f(pa) \in V\}\\
    =&\ \{(s,p) \mid (F_Sf)(s,p) \in \{(s,q) \mid qa\in V\} \}\\
    =&\ \inverse{(F_Sf)} \lambda_Y^a V.
  \end{align*}
  Moreover, $\lambda_X^a$ preserves arbitrary meets since
  \[
  \lambda_X^a \bigcap_{i\in I} U_i = S \times \{p \mid pa \in \bigcap_{i} U_i \} = \bigcap_{i} S \times \{p \mid pa\in U_i\}
  = \bigcap_{i} \lambda_X^a U_i.
  \]
  Likewise, we derive the related facts for $\lambda_X^s$.
  \begin{align*}
    \lambda_X^s\inverse{f}V =&\ \{(s,p) \mid p\in X^\act\}\\
    =&\ \{(s,p) \mid (s,f\circ p)\in \{s\} \times Y^\act\}\\
    =&\ \{(s,p) \mid (F_Sf)(s,p) \in \{s\} \times Y^\act\}\\
    =&\ \inverse{(F_Sf)} \lambda_Y^s V.
  \end{align*}
  Moreover, $\lambda_X^s$ preserves arbitrary meets since
  \[
  \lambda_X^s \bigcap_{i\in I} U_i = \{s\} \times X^\act = \bigcap_{i} ( \{s\} \times X^\act) = \bigcap_{i} \lambda_X^s U_i.
  \]
\end{proof}
As a result, thanks to Proposition~\ref{prop:EMCAT_indexedcat}, the mappings $\lambda^a,\lambda^s$ are also indexed morphisms of type $\contrapower\forget{\_} \rTo \contrapower\forget{\_}\bar F_s$. Recall the determinised systems
\[
\fpower X \rTo^{(o',\alpha')} F_S(\fpower X) \in \emcat{\fpower},
\]
which is induced by an LTS $X \rTo^\alpha (\fpower X)^\act$. We are now ready to derive the action and observation modalities on the determinised system.
\begin{restatable}{theorem}{EMCatsmod}\label{thm:EMCatsMod}
Consider the compositions for $\heartsuit\in\{a,s\}$:
\[
\contrapower \forget{ \fpower X} \rTo^{\lambda_{\fpower X}^\heartsuit} \contrapower \forget{ \bar F_S \fpower X} \rTo^{\inverse{\forget{(o',\alpha')}}} \contrapower \forget{ \fpower X}.\]
Let $\mathbb U$ be a predicate on $\fpower X$. Then, $\inverse{\forget{(o',\alpha')}} \lambda^a_{\fpower X} \mathbb U=\{U \subseteq X \mid U \step a U_\alpha \implies U_\alpha \in \mathbb U\}$ and $\inverse{\forget{(o',\alpha')}} \lambda^s_{\fpower X} \mathbb U=\{U \subseteq X \mid U \term s\}$.
\end{restatable}
\begin{proof}
Let $\mathbb U\subseteq \fpower X$. Then we derive the modality for action transition:
  \begin{align*}
  \inverse{\forget{(o',\alpha')}} \lambda^a_{\fpower X} \mathbb U
   =&\ \inverse{(o',\alpha')} \{(s',p)\in S\times (\fpower X)^\act \mid p(a)\in\mathbb U\}\\
   =&\ \{U\subseteq X \mid (o',\alpha') U \in \{(s',p)\in S\times (\fpower X)^\act \mid p(a)\in\mathbb U\} \}\\
   =&\ \{U \subseteq X \mid \alpha' (U) \in \mathbb U\}\\
   =&\ \{U \subseteq X \mid U \step a U_\alpha \implies U_\alpha \in \mathbb U\}.
\end{align*}
Likewise, we derive the modalities for observations living in $S$:
\begin{align*}
  \inverse{\forget{(o',\alpha')}} \lambda^s_{\fpower X} \mathbb U
   =&\ \inverse{(o',\alpha')} \{(s,p) \mid p\in (\fpower X)^\act \}\\
   =&\ \{U\subseteq X \mid (o',\alpha') U \in \{(s,p) \mid p\in (\fpower X)^\act \} \}\\
   =&\ \{U \subseteq X \mid o' (U) =s\}\\
   =&\ \{U \subseteq X \mid U \term s\}.&& \quad\qed
\end{align*}
\end{proof}
Furthermore, a well-known application of the Linton's theorem when the monad $T$ is finitary is that the category of Eilenberg-Moore algebras is cocomplete. So, in particular, the finite powerset monad $\fpower$ is finitary, thus, $\emcat{\fpower}$ has all coequalisers. Therefore, \ref{assum:quotient} is satisfied.
\end{document}